\documentclass[11pt,letterpaper]{article}

\usepackage{amsmath,amsthm,amsfonts,bbm}
\usepackage{color}
\usepackage[margin=1in]{geometry}
\usepackage[utf8]{inputenc}
\usepackage{lmodern}
\usepackage{mathtools}
\usepackage{newtxmath}      

\usepackage{hyperref}

\newtheorem{theorem}{Theorem}

\newtheorem{lemma}[theorem]{Lemma}

\newtheorem{corollary}[theorem]{Corollary}

\DeclarePairedDelimiter{\norm}{\lVert}{\rVert}

\newcommand{\cL}{\mathcal L}

\newcommand{\R}{\mathbb{R}}

\newcommand{\Off}{\text{Off}}

\newcommand{\rt}{\vvmathbb{r}}

\newcommand{\Ind}{\vvmathbb 1}
\let\1\Ind

\makeatletter
\renewcommand\paragraph{\@startsection{paragraph}{4}{\z@}%
                                    {1.5ex \@plus1ex \@minus.2ex}%
                                    {-1em}%
                                    {\normalfont\normalsize\bfseries}}

\makeatother

\allowdisplaybreaks

\title{Online metric allocation}
\author{Nikhil Bansal\thanks{University of Michigan, Ann Arbor. \texttt{bansal@gmail.com}} \and Christian Coester\thanks{Tel Aviv University, Israel. \texttt{christian.coester@gmail.com}}}
\date{}

\begin{document}

\hypersetup{pageanchor=false}
\maketitle

\begin{abstract}

\vspace{2mm}
 
    We introduce a natural online allocation problem that connects several of the most fundamental problems in online optimization. Let $M$ be an $n$-point metric space. Consider a resource that can be allocated in arbitrary fractions to the points of $M$. At each time $t$, a convex monotone cost function $c_t\colon[0,1]\to\R_+$ appears at some point $r_t\in M$. In response, an algorithm may change the allocation of the resource, paying movement cost as determined by the metric and service cost $c_t(x_{r_t})$, where $x_{r_t}$ is the fraction of the resource at $r_t$ at the end of time $t$. For example, when the cost functions are $c_t(x)=\alpha x$, this is equivalent to randomized MTS, and when the cost functions are $c_t(x)=\infty\cdot\1_{x<1/k}$, this is equivalent to fractional $k$-server.
    
    \smallskip
    
    We give an $O(\log n)$-competitive algorithm for weighted star metrics. Due to the generality of allowed cost functions, classical multiplicative update algorithms do not work for the metric allocation problem. A key idea of our algorithm is to decouple the rate at which  a variable is updated from its value, resulting in interesting new dynamics. This can be viewed as running mirror descent with a \emph{time-varying} regularizer, and we use this perspective to further refine the guarantees of our algorithm. The standard analysis techniques run into multiple complications when the regularizer is time-varying, and we show how to overcome these issues by making various modifications to the default potential function.
    
    \smallskip
    
    We also consider the problem when cost functions are allowed to be non-convex. In this case, we give tight bounds of $\Theta(n)$ on tree metrics, which imply  deterministic and randomized competitive ratios of $O(n^2)$ and $O(n\log n)$ respectively on arbitrary metrics. Our algorithm is based on an $\ell_2^2$-regularizer.
\end{abstract}

\thispagestyle{empty}
\clearpage
\newpage
\hypersetup{pageanchor=true}
\setcounter{page}{1}
\section{Introduction}

We consider the following online {\em metric allocation} problem (MAP). There is an underlying metric space $M$ on $n$ points with distances $d(i,j)$
between points $i$ and $j$.
An algorithm maintains an allocation of a resource to the points of $M$, represented by a point in the simplex $\Delta=\{x\in \R_+^M\mid\sum_{i\in M}x_i=1\}$. 
At each time step $t$, tasks arrive at the points of the metric space. The task at $i\in M$ is specified by a non-increasing and convex cost function $c_{t,i}: [0,1] \rightarrow \R_+$, which describes the cost of completing the task as function of the amount of resource available at $i$.\footnote{The case when a task appears at only one point $r_t$ at time $t$, i.e., $c_{t,i}=0$ for $i\ne r_t$, is equivalent. See Appendix~\ref{app:simplified}.}
Upon receiving the tasks, the algorithm can modify its previous allocation $x(t-1)\in\Delta$ to $x(t)\in\Delta$. It then incurs a {\em service cost}
 $c_t(x(t))=\sum_{i\in M} c_{t,i}(x_i(t))$ and a {\em movement cost} of modifying $x(t-1)$ to $x(t)$ according to the distances given by the metric (i.e., sending an $\epsilon$ fraction of the resource from $i$ to $j$ incurs cost $\epsilon\cdot d(i,j)$). Later, we will also consider the case when cost functions are allowed to be non-convex.
 
 The problem has a natural motivation. For example, the resource may represent workers that can be allocated to various locations. At step $t$, one could transfer extra workers to locations with high cost to execute the tasks more efficiently. This example also illustrates the motivation for cost functions being non-increasing (having more resources can only help) and convex (adding extra resources has diminishing returns).
 
More importantly, MAP also generalizes and is closely related to several fundamental and well-studied problems in online computation.
 \begin{itemize}
     \item \textbf{Metrical Task Systems.} In the Metrical Task Systems (MTS) problem there is a metric space $(M,d)$ on $n$ points, and the algorithm resides at some point in $M$ at any time. At time $t$, a cost vector $\alpha_t\in\R_+^M$ arrives.
     The algorithm can then move from its old location $i_{t-1}\in M$ to a new $i_t\in M$, paying movement cost $d(i_{t-1},i_t)$ and service cost $\alpha_{t,i_t}$.
     
     The state of a \emph{randomized} algorithm for MTS is given by a probability distribution $x(t)=(x_1(t),\ldots,x_n(t))$ on the $n$ points. Its expected service cost at time $t$ is given by $\sum_i\alpha_{t,i},x_i(t)$ and its expected movement cost is measured just like in MAP. Thus, randomized MTS is the special case of MAP when cost functions are of the form $c_{t,i}(x_i) = \alpha_{t,i}x_i$, i.e., linear and increasing. (We show in Appendix~\ref{app:incrDecr} that MAP with non-decreasing cost functions is equivalent to MAP with non-increasing cost functions, so this is indeed a special case.)
  
 \item \textbf{$\boldsymbol{k}$-server.}
 In the $k$-server problem, there is a metric space $(M,d)$ and $k$ servers that reside at points of $M$. At time $t$, some point $r_t$ is requested,
which must be served by moving a server to $r_t$. The goal is to minimize the total movement cost.
The \emph{fractional} $k$-server problem is the relaxation\footnote{All known randomized $k$-server algorithms with polylogarithmic competitive ratios are based on this relaxation.
} of the randomized $k$-server problem where points can have a fractional server mass. 
A request at $r_t$ is served by having a server mass of at least $1$ at $r_t$.

Observe that fractional $k$-server is the special case of MAP with cost functions $c_{t,i}=0$ for $i\ne r_t$ and $c_{t,r_t}(x_{r_t})= \infty$ if $x_{r_t} < 1/k$ and $0$ if $x_{r_t} \geq 1/k$, by viewing
$kx_i(t)$ as the fractional server mass at location $i$ at time $t$. Note that such $c_{t,i}$ are convex and non-increasing.

 \item \textbf{Convex function chasing.}
 In convex functions chasing, the request at time $t$ is a convex function $f_t: \R^n \rightarrow \R^+\cup\{\infty\}$. 
 The algorithm maintains a point in $\R^n$, and given $f_t$ it can move from its old position $x(t-1)\in \R^n$ to a new $x(t)\in\R^n$, incurring cost $\|x(t)-x(t-1)\| +  f_t(x(t))$.
 
 MAP on a star metric captures a special case of convex function chasing where cost functions are supported on the unit simplex and have separable form (i.e.,~$f_t(x)= \sum_{i=1}^n f_{ti}(x_i)$ with $f_{ti}$ monotone).
 \end{itemize}
 
We remark that the $k$-server problem on an $n$-point metric space is also a certain special case of MTS on a metric space with $N=\binom{n}{k}$ point. But as the competitive ratio of MTS depends on $N$, this
does not give any interesting bounds for $k$-server.
In contrast, MAP generalizes both fractional $k$-server and randomized MTS in the \emph{same} metric space.

\subsection{Understanding uniform metric spaces}\label{sec:understanding}
Over the years there has been remarkable progress in obtaining poly-logarithmic guarantees for randomized MTS \cite{BlumKRS00,Sei99,Bar96,BartalBBT97,FiatM03,BansalBN10MTS,BCLL19,CoesterL19} and randomized $k$-server \cite{CoteMP08,BansalBN10Towards,BansalBN12,BansalBMN15,BCLLM18,Lee18,BuchbinderGMN19,GuptaKP21} on general metric spaces. The current best bounds on the competitive ratio are $O(\log^2 n)$ for MTS,  and $O(\log^2 k \log n),\, O(\log^3k\log A)$ for $k$-server, where $A$ is the aspect ratio of the metric space.

All these results for general metrics are actually obtained by using the corresponding results on uniform metrics as building blocks.
Roughly speaking, 
they first approximate a general metric space by hierarchically-separated trees (HSTs) (losing a factor of $O(\log n)$ \cite{FRT04}). Then,  
given an algorithm for uniform metrics (or weighted stars) with a certain refined guarantee (that we describe later below), they 
run this algorithm recursively on each internal node of the HST, to decide how to allocate the resource to each subtree of that node.

\paragraph{MAP on uniform metrics.}
What makes MAP particularly interesting is that it poses several new conceptual difficulties and it is already unclear how to obtain a poly-logarithmic competitive ratio on uniform metrics. The problem is the generality of the convex cost functions. 
To understand this better, let us 
first consider the algorithms for randomized MTS and randomized k-server on star metrics.

The key underlying idea for achieving poly-logarithmic guarantees for these problems is that of multiplicative updates. 
For $k$-server, if a location $r$ is requested, the fractional server amount $z_i$ at other locations $i$ is decreased at rate proportional to $1-z_i+\delta$ (i.e.,~the amount of server already missing at $i$ plus a small constant $\delta$).
Similarly for MTS, if a cost is incurred at location $r$, then the other locations are increased at rate proportionnally to $x_i+\delta$.

For MAP, it is thus clear that one would like to do some kind of multiplicative update to get a polylog guarantee, but what is unclear is ``multiplicative update with respect to what?".

More concretely, if we model $k$-server and MTS  to MAP (as described above), the $k$-server update is $x_i' \propto  (1/k-x_i) + \delta$. This is natural as $1/k$ is a fixed parameter with a special meaning as $c_t(x) = \infty$ for $x_{r_t}<1/k$ and $0$ otherwise. Similarly, for MTS the reason why 
$x_i'\propto (x_i+\delta)$ is natural is that the $c_{t,i}$ always have $x$-intercept at $0$. In contrast, cost functions in MAP are lacking such an intrinsic \emph{scale}.

\paragraph{The difficulty due to scale-freeness.}
We give an instructive example to illustrate the difficulty due to this lack of {\em scale}. Understanding this example will be useful when we describe our ideas later to address these difficulties.

\paragraph{Example.}
We saw above how to model $k$-server via MAP by interpreting $z_i:=k x_i$ as the amount of server mass at $i$, and a request to point $i$ corresponds to the cost function $c_t(x)=\infty\cdot\1_{\{x_i<1/k\}}$. However, this correspondence between the server mass $z_i$ and the variable $x_i$ is quite arbitrary. 

A different but completely equivalent way of modeling $k$-server via MAP is to choose any \emph{offset} vector $a\in[0,1]^n$ with $s:=1-\sum_i a_i>0$, and interpret $z_i:=k\cdot(x_i-a_i)/s$ as the server mass at $i$. Then, a request to page $i$ corresponds to the cost function $c_t(x)=\infty\cdot\1_{\{x_i<a_i+s/k\}}$, and we can additionally intersperse cost functions $c_t(x)=\infty\cdot\1_{\{x_j<a_j\}}$ for each $j$ to ensure that $z_j\ge0$. 
In other words, an adversary can simulate a $k$-server request sequence 
in various \emph{regions} of the simplex and at different \emph{scales}. 

Thus, the challenge for an algorithm is to find out the ``active region'' and ``active scale''. 
What is even worse is that the adversary can keep changing this region and scale arbitrarily over time.
So any online algorithm for MAP needs to learn this region dynamically and determine how to do multiplicative updates with respect to the scale and offset of the current region.

\noindent {\bf Remark.}
At a higher level, the difficulty of learning an ``active region'' is also related to the difficulty in obtaining a polylog$(k)$-competitive algorithm for the $k$-server problem in arbitrary metrics. The reason for the $\log(n)$ term 
in the current bounds 
is due to an HST embedding of the entire metric space of $n$ points. However, an adversary will typically play only in a ``region'' of poly$(k)$ many points (which may change over time), so an algorithm should dynamically learn this region and adapt its strategy accordingly. A step in this direction was made in~\cite{BCLLM18}, where a dynamic HST embedding was used to obtain an $O(\log^3k\log A)$-competitive algorithm for metric spaces with aspect ratio $A$.

\subsection{Results and techniques}

For convex cost functions, we obtain the following tight bound on weighted star metrics.
\begin{theorem}\label{thm:non-refined}
There is an $O(\log n)$-competitive algorithm for MAP on weighted star metrics.
\end{theorem}

This bound is the best possible by the $\Omega(\log n)$ lower bound for the special case of randomized MTS \cite{BLS92}. Our algorithm here is deterministic as randomization does not help here\footnote{Any randomized algorithm for MAP can be derandomized by tracking its expected location. As cost functions are convex, this can only decrease the algorithm's cost.}.

To show Theorem \ref{thm:non-refined}, a key new idea is to handle the  scale-freeness of MAP as described above. We start by decoupling the position $x_i$ and its rate of change $x_i'$ by using separate {\em rate} variables $\rho_i$ that determine $x_i'$. The question then is how to update $\rho_i$. This will be driven by trying to learn the active region and scale.
We give a more detailed overview in Section \ref{sec:firstAlgo}.
To illustrate the main ideas without too many technical details, we will prove Theorem~\ref{thm:non-refined} first for the case of unweighted stars (i.e., uniform metrics).

Theorem \ref{thm:non-refined} 
has an interesting consequence for the natural case convex function chasing described above, as it gives a competitive ratio logarithmic in $n$. For convex function chasing in general, there is an $\Omega(\sqrt{n})$  lower bound \cite{FriedmanL93,BubeckKLLS20}.

\paragraph{Refined guarantees.}
A key first step in extending the results for MTS and $k$-server from star metrics to HSTs is to obtain a more refined  algorithm that has $(1+\epsilon)$-competitive service cost and 
$\text{poly}(\log n,1/\epsilon)$-competitive  movement cost, for $\epsilon \approx 1/\log n$.
 Here, we say that an algorithm has \emph{$\alpha$-competitive service cost} and \emph{$\beta$-competitive movement cost} if, up to some fixed additive constant, its service cost is at most $\alpha$ times the {\em total} (movement plus service) offline cost and its movement cost is at most $\beta$ times the total offline cost.

The reason is that when the algorithm for uniform metrics is used recursively to obtain an algorithm on HSTs, then roughly, the service cost guarantee multiplies across levels and the movement cost guarantee increases additively. See, e.g.,~\cite{BartalBBT97,BansalBMN15}
for more details.

In Section~\ref{sec:refined}, we give an improved algorithm for weighted stars with such refined guarantees. 
\begin{theorem}\label{thm:refined}
	For any $\epsilon>0$, there exists an algorithm for MAP on weighted stars with $(1+\epsilon)$-competitive service cost and $O\left(\frac{1}{\epsilon}+\log n\right)$-competitive movement cost.
\end{theorem}

We 
believe that the refined guarantees in Theorem~\ref{thm:refined} are an important step towards solving MAP on general metrics, and we leave the resolution of this question as an intriguing open problem. 
In Section~\ref{sec:conclusion}, we discuss why the MAP perspective might also help to improve the current bounds on the competitive ratio for $k$-server. 

\paragraph{Time-varying regularization.} Achieving $(1+\epsilon)$-competitive service cost requires additional ideas beyond those needed for Theorem \ref{thm:non-refined}. To describe our algorithm achieving Theorem~\ref{thm:refined}, we will switch in Section~\ref{sec:refined} to the view of \emph{regularization} and \emph{online mirror descent}, a framework that has been applied successfully to various online problems recently~\cite{BCN14,BCLLM18,BCLL19,BuchbinderGMN19,CoesterL19}, including recent breakthrough results for MTS and $k$-server.

However, while these previous works are based on a \emph{static} regularization function, our algorithm uses a regularizer that is \emph{time-varying}. This is necessary as the regularizer must adapt to the current scale and region over time.
This leads to substantial technical complications in the analysis. In particular, the default potential function for mirror descent analyses -- the \emph{Bregman divergence} -- is not well-behaved when the offline algorithm moves (actually, it is not even well-defined), and changes of the regularizer lead to uncontrollable changes of the potential. We show how several technical adaptations of the Bregman divergence can be used to obtain a modified potential function that has all the desired properties necessary to carry out the analysis.

\paragraph{Non-convex costs and arbitrary metric spaces.}
In Section~\ref{sec:nonconvex}, we consider the version of MAP where cost functions are allowed to be non-convex. This version of the problem has an exponentially worse competitive ratio:
\begin{theorem}\label{thm:LB}
On any $n$-point metric space, any deterministic algorithm for MAP with non-convex cost functions has competitive ratio at least $\Omega(n)$.
\end{theorem}

On tree metrics, we complement this lower bound with a matching upper bound:
\begin{theorem}\label{thm:nonconvexTrees}
There is an $O(n)$-competitive deterministic algorithm for MAP on tree metrics, even if the cost functions may be non-convex.
\end{theorem}

By known tree embedding techniques, this implies the following result for general metrics:
\begin{corollary}\label{cor:nonconvexGeneral}
There is an $O(n^2)$-competitive deterministic and $O(n\log n)$-competitive randomized algorithm for MAP on arbitrary metric spaces, even if the cost functions are non-convex.
\end{corollary}

\paragraph{$\boldsymbol{\ell_2^2}$-regularization.} Our algorithm achieving the tight guarantee on trees is also based on mirror descent, but again with a crucial difference to previous mirror-descent based online algorithms in the literature. While previous algorithms all used some version of an entropic regularizer, our regularizer is a weighted $\ell_2^2$-norm. Here, again, the Bregman divergence is not suitable as a potential function, but the issues are more fundamentally rooted in the non-convex structure of cost functions, and addressing them with tweaks to the Bregman divergence seems unlikely to work. Instead, our analysis uses two different potential functions, one of which resembles ideas of ``weighted depth potentials'' used in \cite{BCLLM18,BCLL19} and the other of which is a kind of ``one-sided matching''.

\subsection{Other related work}
A closely related  convex allocation problem was introduced in \cite{BansalBN10Towards} and subsequently with additional ideas in \cite{BansalBMN15} to obtain the first poly-logarithmic guarantees for $k$-server  on general metrics. This problem corresponds to the special case of MAP on a uniform metric when the convex functions $c_{t,i}$ are piece-wise linear determined by values $c_{t,i}(j)$ at $j=0,1/k,2/k,\ldots,1$. 
Roughly, this models the cost of having $j$ servers at location $i$. They gave an algorithm with $(1+\epsilon)$-competitive service cost $O(\log k/\epsilon)$-competitive movement cost.

Notice that this problem does not suffer from scale-freeness as the cost functions have scale $1/k$. In fact, MAP corresponds to this problem as $k\to\infty$. However, in this case the $O(\log k)$ bound of \cite{BansalBN10Towards} does not give anything useful.

The convex function chasing problem was introduced in \cite{FriedmanL93} and has seen remarkable progress recently. The first competitive algorithm for arbitrary dimension $n$ was given in \cite{BubeckLLS19}, and an optimal competitive ratio of $O(n)$ was shown in \cite{Sellke20,ArgueGGT21}.

\subsection{Organization}
In the next section, we first define an equivalent version of MAP that will be easier to work with. We will give a first algorithm for unweighted stars in Section~\ref{sec:firstAlgo}, where we also describe the ideas to overcome scale-freeness, which requires decoupling the update rate of a variable from its value. In Section~\ref{sec:refined}, we give a modified algorithm for weighted stars, which is induced by mirror descent with a time-varying regularizer. We show how to adapt mirror descent analysis techniques to the complications resulting from the time-varying regularizer and prove Theorems~\ref{thm:non-refined} and \ref{thm:refined}. In Section~\ref{sec:nonconvex}, we prove the upper and lower bounds for non-convex functions on general metrics. In particular, this section includes our $\ell_2^2$-regularized tight algorithm for trees. Section~\ref{sec:conclusion} concludes with a discussion about extensions of MAP to general metrics and potential impact on the $k$-server problem.

\section{Preliminaries}\label{sec:prelim}
For $a\in\R$, we will often write $[a]_+:=\max\{a,0\}$. A metric space $M$ is called a \emph{weighted star} if there are weights $w_i>0$ for $i\in M$ and the distance between two points $i\ne j$ is given by $d(i,j)=w_i+w_j$.

\paragraph{Continuous-time model and simplified cost functions.}

It will be more convenient to work with the following continuous-time version of MAP.  Instead of cost functions being revealed at discrete times $t=1,2,\dots$, we think of cost functions $c_t$ arriving continuously over time, and that $c_t$ changes only finitely many times.  At any time $t \in [0,\infty)$, the algorithm maintains a point $x(t) \in \Delta$, where $\Delta:=\{x\in\R_+^M\colon \sum_{i\in M}x_i=1\}$, and 
the dynamics of the algorithm is specified by the derivative $x'(t)$ at each time $t$. The movement cost and the service cost are given by $\int_0^\infty \norm{x'(t)} dt \text{ and }  \int_0^\infty c_t(x(t)) dt$ respectively, where the norm $\norm{\cdot}$ is induced by the metric.\footnote{In general, $\norm{z}=\min_f \sum_{i,j\in M}f(i,j)d(i,j)$, where the minimum is taken over all flows $f\colon M\times M\to \R_+$ satisfying $z_i=\sum_{j\in M}(f(j,i)-f(i,j))$ for all $i\in M$.}  On a weighted star metric, the norm $\norm{\cdot}$ is given by $\norm{z}:=\sum_i w_i |z_i|$.

Further, we assume that the cost functions $c_t$ are of the simple form $c_t(x)=[s_t-x_{r_t}]_+$ for some  $s_t\in[0,1]$ and $r_t\in M$. In other words, $c_{t,i}$ is non-zero for only a single location $r_t\in M$ at any time, and $c_{t,i}$ is piecewise linear and with slope $-1$ in the area where it is positive (see Figure~\ref{fig:simplified_cost}(a)). 
A useful consequence of this view is that if the service cost $c_t(x(t))$ incurred by the online algorithm is $\alpha_t$, then the (one-dimensional) cost function $c_{t,r_t}$ intercepts the $x$-axis (becomes $0$) at the point $s_t=x_{r_t}(t)+\alpha_t$. The cost of an offline algorithm at $y\in\Delta$ is then given by $[\alpha_t+x_{r_t}(t)-y_{r_t}]_+$. For a cost function $c_t$ of this form, we will say that \emph{$x(t)$ is charged cost $\alpha_t$ at $r_t$}.

To simplify the description and analysis of our algorithms, we will further allow them to decrease $x_i$ to a negative value. In Appendix~\ref{app:simplified} we argue that these assumptions are all without loss of generality.
\section{A first algorithm for uniform metrics}\label{sec:firstAlgo}
We will describe here the $O(\log n)$-competitive algorithm for MAP on uniform metrics (i.e., unweighted stars). The algorithm also extends to the weighted star case, but we will assume that all weights are $1$ to avoid technicalities and focus on the key ideas.
In the next section, we will give an algorithm for weighted stars that additionally achieves near-optimal \emph{refined guarantees}.

\subsection{Overview} Before the formal description, we first give a high-level overview and some intuition behind the ideas needed to handle the difficulties due to scale-freeness.

Fix a time $t$, and let $x = x(t)$ and $y = y(t)$ denote the online and offline position and suppose a cost of $\alpha = \alpha_t$ is received (charged) at point $r = r_t$.  We drop $t$ from now for notational ease, as everything is a function of $t$.  Clearly, an algorithm that wishes to be competitive must necessarily increase $x_r$ (otherwise the offline algorithm can move to some $y$ with $y_r>x_r$ and keep giving such cost functions forever). So, the key question is how to decrease other coordinates $x_i$ for $i\neq r$ to offset the increase of $x_r$ (and maintain $\sum_i x_i=1$).

\paragraph{Separate rate variables and how to update them.}
As discussed in Section \ref{sec:understanding}, due to the scale-free property, the rate of update of $x_i$, denoted $x_i'$, cannot simply be some function of $x_i$, as in the usual multiplicative update based algorithms.
So we maintain separate {\em rate} variables $\rho_i$ (decoupled from $x_i$) that specify the rate at which to reduce $x_i$, i.e., we will have $x_i' = -\rho_i$  (plus a small additive term and suitably normalized, but we ignore this technicality for now). Now the key issue becomes how to update these $\rho_i$ variables themselves? 

Consider the following two scenarios, which suggest two conflicting updates to $\rho_r$, depending upon the offline location $y_r$.

(i) $x_r < y_r$.  Here, the adversary can make us incur the service cost while possibly not paying anything itself. However this is not problematic, as we will increase $x_r$ and hence get closer to $y_r$.  Moreover, decreasing $\rho_r$ is also the right thing to do, as it will prevent us in future from decreasing $x_r$ again too fast and move away from $y_r$ when requests arrive at locations $i\neq r$.

(ii) $y_r < x_r$. Here, increasing $x_r$ is possibly fine, as even though we are moving {\em away} from $y_r$, the offline algorithm is paying a higher service cost than online.  However, decreasing $\rho_r$ would be very bad, as this makes it much harder for online to catch up with the offline position $y_r$ later.

To summarize, in case (i), we should decrease $\rho_i$ and in case (ii), we should leave it unchanged or increase it.  The problem is of course that the algorithm does not know the offline location $y_r$, and hence which option to choose.

\begin{figure}[htbp]
    \centering
    \includegraphics{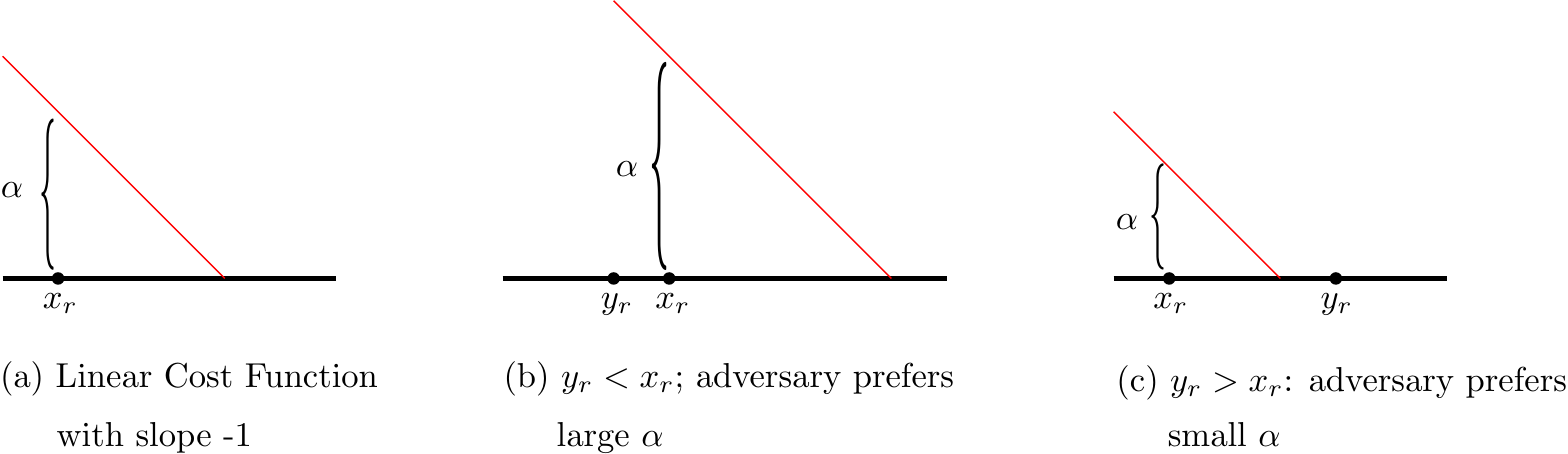}
    \caption{Illustration of cost functions.}
    \label{fig:simplified_cost}
\end{figure}

\paragraph{Indirectly estimating $\boldsymbol{y_r}$.}
We note that even though the algorithm does not know $y_r$, it can reasonably estimate whether $y_r< x_r$ or not, by looking at the structure of requests from the adversary's point of view.  Suppose $y_r < x_r$ (see Figure~\ref{fig:simplified_cost}(b)). In this case, the adversary will want to give us requests with {\em large} $\alpha$; because for small $\alpha$, the offline algorithm pays much higher cost in proportion to that of online; indeed, as $\alpha$ gets larger, the ratio between the offline to online service costs tends to $1$. 
On the other hand, if $y_r>x_r$ (see Figure \ref{fig:simplified_cost}(c)),  the adversary will tend to give requests so that $x_r + \alpha \leq y_r$ (so offline incurs no service cost at all), and hence keep $\alpha$ {\em small}. 

A problem however is that even though the online algorithm sees $\alpha$ when the request arrive, whether this $\alpha$ is large or small has no intrinsic meaning as this depends on the current scale at which the adversary is giving the instance.
 The final piece to make this idea work is that the algorithm will also try to learn the scale at which the adversary is playing its strategy. We describe this next.
 
\paragraph{A region estimate.}

At each time $t$, and for each point $i$, we maintain a real number $b_i\geq 0$ where $b_i \geq x_i$.
Intuitively, we can view $b_i$ as an (online) estimate of the ``region'' where the adversary is playing the strategy, and we set the rate variable $\rho_i = b_i -x_i$. The observations above now suggest the following algorithm.
For an incoming request at point $r$, 

(i) if $x_r+\alpha \leq  b_r$ (this corresponds to small $\alpha$ in the discussion above, which suggested that $y_r > x_r$) the algorithm increases $x_r$ and decreases $b_r$ at roughly the same rate (and hence decreases $\rho_r= b_r-x_r$ at roughly twice the rate), and

(ii) if $x_r +\alpha > b_r$ (this corresponds to $\alpha$ being large, which suggested that $y_r < x_r$), the algorithm increases both $b_r$ and $x_r$ at roughly the same rate (and $\rho_r$ only increases slowly).
Note that even though the location $r$ is incurring a service cost, we do not necessarily panic and decrease $\rho_r$.

For other points $i\neq r$, $b_i$ stays fixed. So as $x_i$ decreases, this corresponds to increasing the rate $\rho_i$. 
This step is analogous to multiplicative updates, but where the update rate is given by distance of $x_i$ from $b_i$, where $b_i$ itself might change over time. We also call $b_i$ the ``baseline''.

A complication (in the analysis) will be that the $b_i$ themselves are only estimates and could be wrong. E.g., even if $b_i$ is accurate at some given time, the offline algorithm can move $y_i$ somewhere far at the next step, and start issuing requests in that region. The current $b_i$ would be completely off now and the algorithm may make wrong moves. What will help here in the analysis is that algorithm is quickly trying learn the new $b_i$.

\begin{figure}
    \centering
    \includegraphics{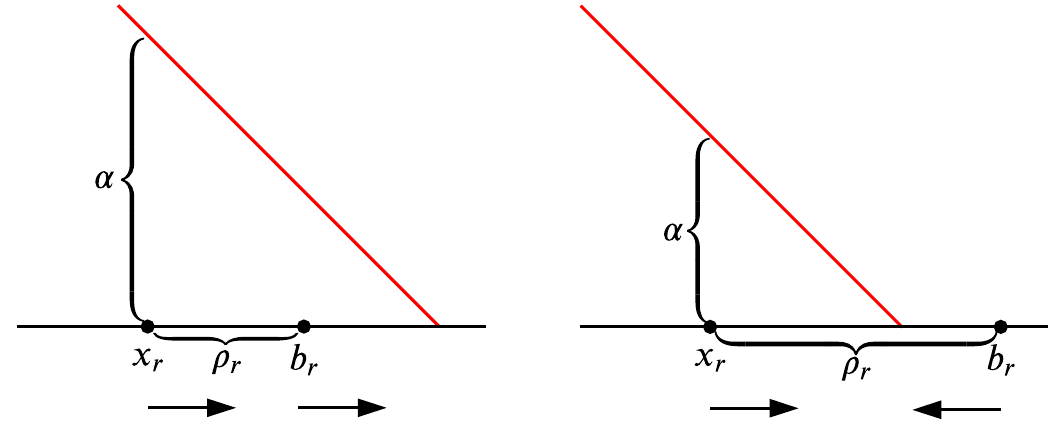}
    \caption{Illustration of the update rule for $\rho_r$. If $\alpha>\rho_r$ (left), then $x_r$ and $b_r$ increase and $\rho_r$ changes slowly. If $\alpha<\rho_r$ (right), then $x_r$ and $b_r$ move towards each other and $\rho_r$ decreases. }
    \label{fig:b_update}
\end{figure}

\subsection{Formal description of the algorithm}
At any instantaneous time $t$, the algorithm maintains a point $x(t) \in \Delta$. In addition, it also maintains a point $b(t) \in \R^M$, where $b_i(t) \geq x_i(t)$ for all $i \in M$. 
Let $\rho_i(t)= b_i(t) -x_i(t)$, for each $i \in M$. We specify the formal algorithm by describing how it updates the points $x(t)$ and $b(t)$ in response to a cost function $c_t$. Again, we drop $t$ from the notation hereafter.

Suppose the cost function at a given time charges cost $\alpha$ at point $r$. Then, we increase 
$x_{r}$ at rate $\alpha$  and simultaneously decrease all $x_i$ (including $x_{r}$) at rate
\[ -\alpha \cdot \frac{\rho_i+\delta S}{2S}
 \]
 where $\delta = 1/n$ and $S:=\sum_i\rho_i$. Intuitively, we can think of $S$ as the current \emph{scale} of the problem (which is changing over time).

Then the overall update of $x$ can be summarized as
\begin{align}
\label{eq:xupdate}
 x_i'=\alpha\left( \1_{i=r}-\frac{\rho_i+\delta S}{2 S}\right)\qquad \text{for all $i \in M$.} 
\end{align}
The baseline vector $b$ is updated as follows:

(i) For $i \neq r$, $b_i$ stays fixed. 

(ii) For $i=r$, if $x_r + \alpha > b_r$ (or equivalently $\alpha> \rho_r$), then $b_r' = \alpha$. 

(iii)
For $i=r$, if $x_r+\alpha \leq b_r$ (or equivalently $\alpha \leq \rho_r)$, then $b_r' = -\alpha$.  \footnote{Strictly speaking, if $x_r+\alpha=b_r$ both rules (ii) and (iii) apply simultaneously and $b_r$ stays fixed.}

See Figure \ref{fig:b_update} for an illustration. 
Notice that the update rules for $x_i'$ and $b_i'$ also ensure that $b_i \geq x_i$. Finally, using that $\rho_i = b_i-x_i$ and  
writing compactly, this
gives the following update rule for $\rho_i$.
\begin{equation}
 \rho_i'=\alpha\left(\frac{\rho_i+\delta S}{2 S} - 2\cdot \1_{i=r \text{ and }  \rho_r>\alpha}\right).  
\end{equation}
This completes the description of the algorithm. 

\paragraph{Feasibility.}

Notice that
 $\sum_i  (\rho_i + \delta S)/ (2 S) = 1$
and hence $\sum_i x_i'=0$ and the update for $x$ maintains that $\sum_i x_i=1$.
In the algorithm description above, we do not explicitly enforce that $x_i \geq 0$. As we show in Appendix~\ref{app:simplified}, allowing the online algorithm to decrease $x_i$ to negative values is without loss of generality. It is possible to enforce this directly in the algorithm, but this would make the notation more cumbersome. 

\subsection{Analysis sketch}
The analysis of the algorithm is based on potential functions. Specifically, we define a (bounded) potential 
$\Theta$ that is a function of the online and offline states, and show that at any time $t$ it satisfies
\begin{equation}
\label{eq:pot}
    \text{On}' + \Theta' \leq O(\log n) \cdot \text{Off}',
\end{equation} 
where $\text{On}'$ (resp.~$\text{Off}'$) denote the change in cost of the online (resp.~offline) algorithm, and $\Theta'$ is the change in the potential.
The potential function $\Theta$ consists of two parts defined as
\begin{align*}
 &\text{(Primary potential)} &&P := \sum_{i\colon x_i \geq y_i} (b_i-y_i)\log\frac{(1+\delta)(b_i-y_i)}{\rho_i+\delta  (b_i-y_i)}. \\
&\text{(Scale-mismatch potential)}  &&Q := \sum_{i} \left[\rho_i + 2(x_i - y_i)\right]_+.
\end{align*}
where $y \in \Delta$ is the position of the offline algorithm. The overall potential is given by 
\[\Theta = 12P+6Q.\]
It is easy to verify, by taking derivative with respect to $y_i$ and using $\delta =1/n$, that $\Theta$ is $O(\log n)$-Lipschitz in $y_i$. Thus the potential increases by at most $O(\log n)$ times the offline movement cost when $y$ changes. It therefore suffices to show \eqref{eq:pot} for the case that $y$ stays fixed and only the online algorithm moves.

Recall that $S = \sum_i \rho_i$ and define $L:=\sum_i(x_i-y_i)_+$.  Since $P$ and $Q$ change continuously and their derivatives exist for almost all times, we can ignore times when the derivatives do not exist. We prove the following:
\begin{lemma}[Change of the primary potential]\label{lem:primaryChange}
When $y$ is fixed and the online algorithm moves, the change of $P$ is bounded by
\[
    P'\le O(\log n)[\alpha+x_r-y_r]_+ \,-\,\frac{\alpha}{2}\min\left\{L/S\,,\, 1\right\}.
\]
\end{lemma}
\begin{proof}
We  consider two cases depending on whether $x_r < y_r$ or not.
\paragraph{ If $\boldsymbol{x_r<y_r}$.}  In this case, the term for point $r$ does not contribute to $P$ (as the sum in $P$ is only over points $i$ such that $x_i \geq y_i$). For $i\neq r$, $b_i$ does not change, so the change of $P$ only results from the change of the $\rho_i$. Thus,
\[
P'= -\frac{\alpha}{2}\sum_{i\colon x_i\ge y_i}\frac{b_i-y_i}{\rho_i+\delta(b_i-y_i)}\frac{\rho_i+\delta S}{S}.\]
If $b_i-y_i\le S$ for all $i$ with $x_i\ge y_i$, then the sum can be bounded as
\begin{align*}
\sum_{i\colon x_i\ge y_i}\frac{b_i-y_i}{\rho_i+\delta(b_i-y_i)}\frac{\rho_i+\delta S}{S} &\ge \sum_{i\colon x_i\ge y_i}\frac{b_i-y_i}{S} \ge \sum_{i\colon x_i\ge y_i}\frac{x_i-y_i}{S} =\frac{L}{S}.
\end{align*}
Otherwise, there exists $j$ with $x_j\ge y_j$ and $b_j-y_j>S$. Then
\begin{align*}
\sum_{i\colon x_i\ge y_i}\frac{b_i-y_i}{\rho_i+\delta(b_i-y_i)}\frac{\rho_i+\delta S}{S} &\ge \left(\frac{\rho_j}{b_j-y_j}+\delta\right)^{-1}\left(\frac{\rho_j}{S}+\delta\right) > 1.
\end{align*}
In either case, the sum is at least $\min\left\{\frac{L}{S},1\right\}$, so we get
\begin{align}
P'&\le -\frac{\alpha}{2}\min\left\{\frac{L}{S}, 1\right\},\label{eq:PsimpleBound}
\end{align}
which proves the lemma for the case $x_r<y_r$.

\paragraph{If $\boldsymbol{x_r\ge y_r}$.} Here the change of $P$ has an additional contribution beyond the bound \eqref{eq:PsimpleBound} due to point $r$.
The contribution can be due to decrease of $\rho_r$ at rate $2\alpha\cdot\1_{\rho_r>\alpha}$ and the change of $b_r$.

The decrease of $\rho_r$ at rate $2\cdot\1_{\rho_r>\alpha}$ causes change of $P$ at rate
\begin{align}
2\alpha\cdot\1_{\rho_r > \alpha} \frac{b_r-y_r}{\rho_r+\delta(b_r-y_r)}
&\le 2\alpha\cdot\1_{\rho_r>\alpha} \frac{\rho_r+x_r-y_r}{\rho_r}
\le 2\cdot[\alpha + x_r-y_r]_+,\label{eq:PExtraRho}
\end{align}
where we used that $x_r\ge y_r$.

Now consider the effect of the change of $b_r$. If $x_r+\alpha\le b_r$, then $b_r$ is decreasing and this can only help to decrease $P$, which we can ignore. If $x_r+\alpha>b_r$, then $b_r$ increases at rate $\alpha$. The increase of $b_r$ in the denominator helps to decrease $P$ (which we can ignore), so the additional increase of $P$ due to increasing $b_r$ at rate $\alpha$ is at most at rate
\begin{align}
\alpha\left(\log\frac{(1+\delta)(b_r-y_r)}{\rho_r+\delta(b_r-y_r)}+1\right)\le  \alpha\left(\log\frac{1+\delta}{\delta}+1\right)\le  [\alpha+x_r-y_r]_+ \cdot O(\log n),\label{eq:PExtrab}
\end{align}
where we used again that $x_r\ge y_r$.
We conclude the lemma by adding \eqref{eq:PsimpleBound} and the additional contributions \eqref{eq:PExtraRho} and \eqref{eq:PExtrab}.
\end{proof}
\medskip

In this bound on $P'$ in Lemma \ref{lem:primaryChange}, the positive term can be charged against the offline service cost. If $S=O(L)$, then the negative term can be used to pay for the online cost. However, if the scale $S$ is much larger than $L$, then this term is $P'$ (i.e.~reduction in potential) may be negligibly small. 
Notice that  $S=\sum_i \rho_i = \sum_i (b_i -x_i)$ and $L = \sum_i (x_i-y_i)_+ = \sum_i |x_i-y_i|/2$ as both $x,y$ lie on the simplex, so essentially $S \gg L$ means that our estimate of the scale is far off from the true scale.
In this case, the scale-mismatch potential comes into play, whose change we bound below.

\begin{lemma}[Change of the scale-mismatch potential]\label{lem:scaleChange}
When $y$ is fixed and the online algorithm moves, the change of $Q$ is bounded by
\[
    Q'\le 2\alpha\cdot \1_{y_r\le x_r+\alpha/2}\,-\,\left[ \frac{1}{2} - \frac{L}{S} \right]_+\alpha.\]
    \end{lemma}
\begin{proof} By the definition of $Q$  and the updates for changes of $x_i$ and $\rho_i$,  we can write the change in $Q$ compactly as
\begin{align*} 
Q'&= 2\alpha\cdot \1_{\rho_r\le \alpha\text{ and }\rho_r\ge 2(y_r-x_r)}-\frac{\alpha}{2}\sum_{i\colon \rho_i>2(y_i-x_i)}\frac{\rho_i+\delta S}{S}\\
&\le 2\alpha\cdot \1_{y_r\le x_r+\alpha/2}-\frac{\alpha}{2} \bigg[1-\sum_{i\colon \rho_i\le2(y_i-x_i)}\frac{\rho_i}{S}\bigg]_+\\
&\le 2\alpha\cdot \1_{y_r\le x_r+\alpha/2}-\alpha \left[\frac{1}{2}-\frac{L}{S}\right]_+.\qedhere
\end{align*}
\end{proof}

Note if $y_r\le x_r+\alpha/2$, then the offline service cost is at least $\alpha/2$, so the positive part in the change of $Q$ is at most 4 times the offline service cost. The negative part can pay for the online cost if $S\gg L$, which is precisely the case not covered by the primary potential. Concretely, we have
\begin{align*}
    \Theta' = 12P'+6Q' & \le -3\alpha\left(\min\left\{\frac{2L}{S}, 2\right\} + \left[1 - \frac{2L}{S}\right]_+\right) + O(\log n)\cdot\Off'  \\
    & \le -3\alpha + O(\log n)\cdot \Off',
\end{align*}
which yields the desired inequality~\eqref{eq:pot} because the online service cost is $\alpha$ and online movement cost is at most $2\alpha$.
\section{\texorpdfstring{$(1+\epsilon)$}{(1+eps)}-competitive service cost on weighted stars}\label{sec:refined}

We now give an improved algorithm for weighted stars, proving Theorem~\ref{thm:non-refined} and the more refined guarantee of Theorem~\ref{thm:refined} that the service cost is $(1+\epsilon)$-competitive.

For the MTS problem, a systematic way to achieve $1$-competitive service cost is via the online mirror descent framework \cite{BCLLM18}. We will describe how one can view our algorithm from the previous section through a mirror descent lens, and how to modify it so that its service cost competitive ratio is nearly $1$. For the MTS problem, the service cost analysis of mirror descent algorithms proceeds in a mostly black box way. However, due to various technical complications arising from the time-varying regularizer, this will not be possible for our analysis, but we will show how analysis techniques can be adapted to achieve our desired result.

\subsection{Online mirror descent}
The online mirror descent framework has been useful to derive optimally competitive algorithms for problems where the state of an algorithm can be described by a point in a convex body (e.g., set cover, $k$-server, MTS \cite{BCN14,BCLLM18,BCLL19,BuchbinderGMN19,CoesterL19}). That is, the algorithm can be described by a path $x\colon[0,\infty)\to \Delta$ for a convex body $\Delta\subset \R^n$, such that $x(t)$ describes the state of the algorithm at time $t$ (in our case, $\Delta$ is the simplex). In the continuous-time mirror descent framework, the dynamics of an algorithm $x$ is specified by a differential equation of the form
\begin{align}
    \nabla^2\Phi(x(t))\cdot x'(t) = f(t) - \lambda(t)\label{eq:MD}
\end{align}
where $\Phi\colon\Delta\to \R$ is a suitable convex function called the \emph{regularizer}, $\nabla^2\Phi(x(t))$ is its Hessian at $x(t)$, $f\colon [0,\infty)\to \R^n$ is called a \emph{control function}, and $\lambda(t)$ is an element of the normal cone of $\Delta$ at $x(t)$, given by
\[N_\Delta(x(t)):=\{\lambda\in\R^n\colon \langle \lambda, y-x(t)\rangle \le 0, \,\forall y\in\Delta\}.\]
Under suitable conditions (which are all satisfied here), the path $x$ is uniquely defined by \eqref{eq:MD} and absolutely continuous in $t$ \cite{BCLLM18}. The equation \eqref{eq:MD} is easiest to read if we imagine that $\nabla^2\Phi(x(t))$ were the identity matrix. Then \eqref{eq:MD} says that $x$ tries to move in direction $f(t)$, and the normal cone element $\lambda$ ensures that $x(t)$ does not leave the body $\Delta$. The case when $\nabla^2\Phi(x(t))$ is different from the identity matrix corresponds to imposing a different (e.g., non-Euclidean) geometry on $\Delta$.

We show in Appendix~\ref{app:MD} that the dynamics \eqref{eq:xupdate} of $x$ defined in the previous section is precisely equivalent to the mirror descent equation~\eqref{eq:MD} when choosing the \emph{time-varying} regularizer
\begin{align}\Phi_t(x):= \sum_i (b_i(t)-x_i+\delta S(t))\log\left(\frac{b_i(t)-x_i}{S(t)}+\delta\right)\label{eq:regUnweighted}
\end{align}
for $S(t):=\sum_i b_i(t)-1$, and the control function $f(t):=\frac{\alpha_t}{b_{r_t}(t)-x_{r_t}(t)+\delta S(t)}e_{r_t}$, where $e_{r_t}$ denotes the $0$-$1$-vector with a $1$ only in the $r_t$-coordinate, and $\alpha_t$ is the cost charged at $r_t$ at time $t$.

Mirror descent analyses typically employ as a potential function the Bregman divergence, defined by
\begin{align*}
    D_\Phi(y\| x):=\Phi(y)-\Phi(x)+\langle \nabla \Phi(x), x-y\rangle.
\end{align*}
In the classical setting where $\Phi$ is static (rather than time-varying), the change of potential when $x$ moves is given by
\begin{align}
    \frac{d}{dt}D_\Phi(y\| x(t)) 
    &= \langle \nabla^2\Phi(x(t))\cdot x'(t), x(t)-y\rangle\nonumber\\
    &= \langle f(t)-\lambda(t), x(t)-y\rangle\nonumber\\
    &\le \langle f(t), x(t)-y\rangle,\label{eq:BregmanBound}
\end{align}
where the second equation uses \eqref{eq:MD} and the inequality uses that $\lambda(t)\in N_\Delta(x(t))$. Due to \eqref{eq:BregmanBound}, it would seem more natural to choose $f(t)=e_{r_t}$, so that the change in potential allows to charge the online service cost $\alpha_t$ to the offline service cost $[\alpha_t+x_{r_t}-y_{r_t}]_+$:
\begin{align}
    \alpha_t+\frac{d}{dt}D_\Phi(y\| x(t)) \le \alpha_t+x_{r_t}-y_{r_t}.\label{eq:SCcleanBound}
\end{align}
For any \emph{static} regularizer $\Phi$, this implies immediately that the service cost is $1$-competitive against the total offline cost provided that $D_\Phi(y\|x)$ is $1$-Lipschitz in the movement of $y$ (so that any possible increase of the potential when the offline location $y$ changes can be charged to the offline movement cost). Usually, $1$-Lipschitzness in $y$ can be enforced easily by scaling $\Phi$ by a suitable factor. However, for the regularizer \eqref{eq:regUnweighted} we have the problem that $\Phi_t(y)$ may not even be well-defined for the offline location $y$, since the argument of the logarithm could be negative when $b_i\ll y_i$.

We not only want the service cost to be close to $1$-competitive, but also that the movement cost is still $O(\log n)$-competitive. This requires a suitable choice of the regularizer $\Phi_t$, which in our case needs to be changing over time. An immediate complication due to this is that any change of $\Phi_t$ can cause an additional change of the potential $D_{\Phi_t}(y\|x)$.

The following list summarizes the main obstacles and desiderata from this discussion.
\begin{enumerate}
    \item $\Phi_t(y)$ should be well-defined for the offline location $y$.
    \item $D_{\Phi_t}(y\|x)$ should be close to $1$-Lipschitz in the movement of $y$.
    \item $\Phi_t$ should be time-varying in a way that facilitates $O(\log n)$-competitive movement cost.
    \item Changes to the potential $D_{\Phi_t}(y\|x)$ due to the time-varying nature of $\Phi_t$ need to be controlled.
\end{enumerate}

\subsection{Adapting mirror descent to time-varying regularization}
To ensure that $\Phi_t(y)$ is well-defined for the offline location $y$, we simply replace the term $b_i(t)-x_i$ inside the logarithm by its positive part $[b_i(t)-x_i]_+$. For the online algorithm, this does not actually change anything as we will ensure that $b_i(t)>x_i(t)$. In particular, $\Phi_t$ will always be strongly convex in a neighborhood of the online location $x(t)$.

Overall, we will use the regularizer
\begin{align}\Phi_t(x):= \frac{1}{\eta}\sum_i w_i(b_i(t)-x_i+\delta S(t))\log\left(\frac{[b_i(t)-x_i]_+}{S(t)}+\delta\right)\label{eq:regWighted}
\end{align}
where $\eta=O\left(\frac{1}{\epsilon}+\log n\right)$ is a scaling factor (which will be proportional to our overall competitive ratio), $\delta>0$ will be chosen suitably small, and $w_i$ are the weights of the weighted star.

From the definition of Bregman divergence, one can calculate that
\begin{align}
    D_{\Phi_t}(y\|x)&=\frac{1}{\eta}\sum_i w_i\left(\left(b_i(t)-y_i+\delta S(t)\right)\log\frac{[b_i(t)-y_i]_++\delta S(t)}{b_i(t)-x_i+\delta S(t)} - x_i + y_i\right).\label{eq:BregmanBad}
\end{align}

Unfortunately, neither the second nor the forth desired property above are satisfied by $D_{\Phi_t}(y\|x)$: When $S\to 0$, the logarithmic term can be unbounded, and therefore $D_{\Phi_t}(y\|x)$ fails to be Lipschitz in the movement of $y$. Moreover, $D_{\Phi_t}$ is not well-behaved when the regularizer changes: While it is possible to control the effect of changing $b_r(t)$ in the summand corresponding to the requested point $r$, the resulting change of $S(t)$ affects all summands and it is unclear what the resulting change of $D_{\Phi_t}(y\|x)$ could be charged against (especially when the weights $w_i$ are very different from each other).

We will address both of these issues by using a modified version of the Bregman divergence as a potential, in which the $S(t)$ terms in the $i$-summand are replaced by $b_i(t)-x_i\land y_i$, where $x_i\land y_i$ denotes the minimum of $x_i$ and $y_i$. With this change, the argument of the logarithm is bounded in $[\delta/(1+\delta), (1+\delta)/\delta]$ (and the potential becomes $1$-Lipschitz in $y$ for sufficiently large $\eta$). Moreover, the change of $b_r(t)$ will then only affect the $r$-summand, and we will be able to account for this. Some additional small tweaks to the potential will be necessary, which we ignore for now.

So, as our potential will be slightly different from the Bregman divergence, the suggested control function $f(t)=e_{r_t}$ will not quite be enough any more to conclude the analogous statement of \eqref{eq:SCcleanBound}. Instead, we will use \[f(t)=\frac{b_{r_t}(t)-x_{r_t}(t)}{b_{r_t}(t)-x_{r_t}(t)+\delta S(t)} e_{r_t},\] which is similar when $\delta S(t)$ is relatively small compared to $b_{r_t}(t)-x_{r_t}(t)$.

As we derive formally in Appendix~\ref{app:MD}, this choice of $\Phi_t$ and $f(t)$ induces the following dynamics of the online location $x$ when the cost function comes at $r=r_t$:
\begin{align}
x_i'(t)&=\eta \frac{b_r(t)-x_r(t)}{w_r}\left(\1_{i=r}-\frac{b_i(t)-x_i(t)+\delta S(t)}{\gamma(t) w_i S(t)}\right)\label{eq:xUpdateNew}
\end{align}
where \[\gamma(t):=\sum_{i}\frac{x_i(t)-b_i(t)+\delta S(t)}{w_i S(t)}.\] In the next subsection, we will describe how to choose the dynamics of $b_i$ so that the third desired property, i.e., $O(\eta)$-competitive movement cost, can be satisfied.

\subsection{Moving the baseline}
Let us assume for a moment that our service cost will be $(1+\epsilon)$-competitive as desired. Then a natural way to achieve $O(\eta)$-competitive movement cost is to ensure that the online movement cost is at most $O(\eta)$ times the online service cost. For MTS, this can be achieved by choosing the movement speed roughly proportionally to $O(\eta)$ times the current service cost~\cite{BCLL19}. Note that the algorithm~\eqref{eq:xupdate} from the previous section had a similar property: Its movement speed was proportional to the service cost $\alpha$. However, for our new update rule \eqref{eq:xUpdateNew}, the movement speed of $x$ is actually independent of $\alpha$, so it is less obvious how to achieve this.

We will still be able to bound the online movement cost against the online service cost, indirectly via the movement of the baseline: As before, the baseline $b_i$ will be changing only when $i=r$ is the requested point. We will choose the update of $b_r$ such that
\begin{alignat}{2}
    &w_rx_r'(t)&&\le O(\eta) w_r|b_r'(t)|\label{eq:brDesired1}\\
    &w_r[b_r'(t)]_+&&\le \alpha_t.\label{eq:brDesired2}
\end{alignat}
for all times $t$. By the first inequality, the movement cost of $x$ is at most $O(\eta)$ times the weighted movement of $b$. By the second inequality (and the fact that the total positive movement of $b_r$ will be equal to the total negative movement of $b_r$ -- up to an additive constant), the weighted movement of $b_r$ is at most the service cost. Combined, this implies that the online movement cost is at most $O(\eta)$ times the online service cost (and thus, if the online service cost is $O(1)$-competitive, the algorithm is overall $O(\eta)$-competitive).

It is straightforward to verify that inequalities \eqref{eq:brDesired1} and \eqref{eq:brDesired2} are achieved by the following update rule of $b$:
\begin{align}
b_i'(t)&=\begin{cases}
\alpha_t/w_r&i=r\text{ and }b_r(t)-x_r(t)\le 2\alpha_t\\
-(b_r(t)-x_r(t))/(2w_r)\quad&i=r\text{ and }b_r(t)-x_r(t) > 2\alpha_t\\
0&i\ne r.
\end{cases}\label{eq:bUpdate}
\end{align}
Qualitatively, this update of $b$ resembles similar ideas that we already motivated in Section~\ref{sec:firstAlgo} (for small $\alpha$ we decrease $b_r$, for large $\alpha$ we increase $b_r$), but quantitatively the new update is more delicate in order to facilitate the rigorous analysis that we will provide in the next subsections.

Note that as before, it is guaranteed that $b_i > x_i$: Whenever $b_r-x_r$ approaches $0$, the rate of movement of $x$ tends to $0$ whereas $b_r$ increases at rate $\frac{\alpha}{w_r}$; and for $i\ne r$, $b_i-x_i$ can only increase.
\subsection{Algorithm}
Let $\delta:=1/\max\{n^2, e^{3/\epsilon}\}$ and $\eta:=(1+\delta)\log\frac{1+\delta}{\delta}=O\left(\frac{1}{\epsilon}+\log n\right)$. For every leaf $i$, we maintain a value $b_i(t)\in(x_i(t),2]$, initialized arbitrarily in this interval. Let $S(t):=\sum_i b_i(t)-1 = \sum_i (b_i(t)-x_i(t))>0$ and $\gamma(t):=\sum_{i}\frac{b_i(t)-x_i(t)+\delta S(t)}{w_i S(t)}$. Hereafter, we drop the dependence on $t$ from the notation. When cost $\alpha>0$ is charged at leaf $r$, then we update $x$ and $b$ as already stated in~\eqref{eq:xUpdateNew} and~\eqref{eq:bUpdate}, i.e.:
\begin{align*}
x_i'&=\eta \frac{b_r-x_r}{w_r}\left(\1_{i=r}-\frac{b_i-x_i+\delta S}{\gamma w_i S}\right)\\
b_i'&=\begin{cases}
\alpha/w_r&i=r\text{ and }b_r-x_r\le 2\alpha\\
-(b_r-x_r)/(2w_r)\quad&i=r\text{ and }b_r-x_r > 2\alpha\\
0&i\ne r.
\end{cases}
\end{align*}

\paragraph{Well-definedness.} First observe that $\sum_i x_i'=0$ by choice of $\gamma$, so $\sum_i x_i=1$ is maintained. As before, to avoid complicating the notation we allow the algorithm to decrease $x_i$ to negative values, which is without loss of generality as per Appendix~\ref{app:simplified}. 

We already argued that $b_i>x_i$ is maintained, which also shows that $S>0$. To see that $b_i\le 2$, note that $b_r$ can increase only if $b_r\le x_r+2\alpha\le 2-x_r\le 2$, where we used that the point $\alpha+x_r$ where cost becomes $0$ lies in $[0,1]$.

\subsection{Service cost analysis}
We use the notation $a\land b:=\min\{a,b\}$.
Our analysis employs a potential function defined as
\begin{align}
D&:= \frac{\beta}{\eta}\sum_i w_i\left[\left(b_i-y_i + \delta(b_i-x_i\land y_i)\right) \log \frac{[b_i-y_i]_+ + \delta(b_i-x_i\land y_i)}{b_i-x_i + \delta(b_i-x_i\land y_i)} - (1-n\delta)x_i + y_i + b_i\right],\label{eq:pseudoBregman}
\end{align}
where $\beta:=1+\frac{2}{\eta}$. So this is similar to the Bregman divergence \eqref{eq:BregmanBad} with $S$ replaced by $b_i-x_i\land y_i$. The factors $\beta$ and $(1-n\delta)$ are needed for technical reasons to handle complications resulting from the fact that this is not actually the Bregman divergence (note that both these factors tend to $1$ as $n\to\infty$) and to allow for accounting of changes of the potential due to changes of $b$. Note also the $b_i$ term at the end, which would not exist in the usual Bregman divergence. Having this term in the potential will allow us in some cases to charge small increases of other parts of the potential to a decrease of $b_r$.

We will sometimes use that $\epsilon\ge 3/\eta$, which follows since $\eta\ge \log(1/\delta)\ge 3/\epsilon$. Some steps of our proof assume that $n$ is sufficiently large (compared to some constant). This is without loss of generality, as we could simply add a constant number of additional leaves to the weighted star.

By the following lemma, the potential is $(1+\epsilon)$-Lipschitz in the offline movement of $y$.
\begin{lemma}\label{lem:offLipschitz}
	For all $i\in M$, $D$ is $(1+\epsilon)w_i$-Lipschitz in $y_i$.
\end{lemma}
\begin{proof}
It suffices to show that if $y_i\notin\{b_i,x_i\}$ (so that the derivative $\partial_{y_i}D$ exists), then
\begin{align*}
    \partial_{y_i}D \in\left[-\frac{\beta w_i}{\eta}\left((1+\delta)\log\frac{1+\delta}{\delta}-1\right), \frac{\beta w_i}{\eta}\left(\log\frac{1+\delta}{\delta}+1\right)\right].
\end{align*}
The lemma then follows since this interval is contained in $[-\beta w_i,\beta w_i]\subseteq [-(1+\epsilon) w_i,(1+\epsilon) w_i]$.

If $y_i>b_i$, then
\[
\partial_{y_i}D = \frac{\beta w_i}{\eta}\left(\log\frac{1+\delta}{\delta}+1\right).
\]
If $x_i<y_i<b_i$, then
\[
\partial_{y_i}D = \frac{\beta w_i}{\eta}\log\frac{(1+\delta)(b_i-x_i)}{b_i-y_i+\delta(b_i-x_i)}\in\left[0,\frac{\beta w_i}{\eta}\log\frac{1+\delta}{\delta}\right].
\]
If $y_i<x_i$, then
\begin{align*}
\partial_{y_i}D &= \frac{\beta w_i}{\eta}\left((1+\delta)\log\frac{b_i-x_i+\delta(b_i-y_i)}{(1+\delta)(b_i-y_i)} -\delta +\delta\frac{(1+\delta)(b_i-y_i)}{b_i-x_i+\delta(b_i-y_i)}\right)\\
&\in\left[\frac{\beta w_i}{\eta}\left((1+\delta)\log\frac{\delta}{1+\delta}+1\right), -\frac{\beta w_i}{\eta}\delta\right].\qedhere
\end{align*}
\end{proof}

Our goal now is to show that whenever the online algorithm is moving and the derivative $D'$ of $D$ with respect to time is well-defined,
\begin{align}
\alpha+D' \le (1+\epsilon)[\alpha+x_r-y_r]_+.\label{eq:wstarGoal}
\end{align}
assuming the offline location $y$ remains unchanged at this time.

Note that this implies $(1+\epsilon)$-competitive service cost since $\alpha$ is the instantaneous online service cost, $[\alpha+x_r-y_r]_+$ is the instantaneous offline service cost, and movement of $y$ increases $D$ by at most $1+\epsilon$ times the offline movement cost by Lemma~\ref{lem:offLipschitz}.

The following lemma bounds the change of $D$ resulting from the change of $x$.
\begin{lemma}
	At almost all times when the online algorithm is moving,
	\begin{align}
	\sum_i x_i' \cdot \partial_{x_i}D&\le -\beta (b_r-x_r)\left(1-2n\delta\right) + \beta(1+\delta)[b_r-y_r]_+.\label{eq:PhiChangeX}
	\end{align}
\end{lemma}
\begin{proof}
	We claim that if $x_i\ne y_i$, then
	\begin{align*}
	\partial_{x_i}D = \frac{\beta w_i}{\eta}\left(- (1-n\delta)  + \frac{[b_i-y_i]_+ + \delta(b_i - x_i\land y_i)}{b_i-x_i + \delta(b_i - x_i\land y_i)} - \delta\ell_i \right),
	\end{align*}
	where
	\begin{align*}
	\ell_i:= \1_{x_i<y_i}\left(\log \frac{[b_i -  y_i]_+ + \delta(b_i - x_i\land y_i)}{b_i- x_i + \delta(b_i - x_i\land y_i)} + 1 - \frac{[b_i -  y_i]_+ + \delta(b_i - x_i\land y_i)}{b_i-x_i + \delta(b_i - x_i\land y_i)}\right)\in[-\eta,0].
	\end{align*}
	One can see this by considering separately the cases $y_i<b_i$ and $y_i\ge b_i$; note that in the latter case, the argument of the logarithm in $D$ is simply $\delta/(1+\delta)$ and thus a constant function of $x_i$.
	
	Therefore,
	\[
	\sum_i x_i' \cdot \partial_{x_i}D= \beta\frac{b_r-x_r}{w_r}\sum_i \left(\frac{b_i-x_i+\delta S}{\gamma S} - \1_{i=r}w_r\right)\left( 1-n\delta- \frac{[b_i- y_i]_+ + \delta(b_i - x_i\land y_i)}{b_i-x_i + \delta(b_i - x_i\land y_i)} + \delta\ell_i\right).
	\]
	We claim that
	\begin{align}
	\sum_i \frac{b_i-x_i+\delta S}{\gamma S}\cdot \left(1-n\delta-\frac{[b_i- y_i]_+ + \delta(b_i - x_i\land y_i)}{b_i-x_i + \delta(b_i - x_i\land y_i)}\right)\le 0.\label{eq:subSum}
	\end{align}
	To see this, we first argue that the left-hand side of \eqref{eq:subSum} is maximized by some $y$ with $y_i\le b_i$ for all $i$. Suppose there is some $i$ with $y_i>b_i$. If $y_i$ is decreased by a small amount while $y_j$ for some $j$ with $y_j<x_j$ is increased by the same amount (so that the new $y$ is still in the simplex), then the left-hand side of \eqref{eq:subSum} gets larger. So indeed, the left-hand side of \eqref{eq:subSum} is maximized when $y_i\le b_i$ for all $i$. In this case, we also have $b_i - x_i\land y_i = \max\{b_i-x_i, b_i-y_i\}\le \max\{\sum_i (b_i-x_i), \sum_i (b_i- y_i)\}=S$. Thus,
	\begin{align*}
	\sum_i &\frac{b_i-x_i+\delta S}{\gamma S}\cdot \left(1-n\delta-\frac{[b_i- y_i]_+ + \delta(b_i - x_i\land y_i)}{b_i-x_i + \delta(b_i - x_i\land y_i)}\right) \\
	&\le \max_{y\colon y_i\ge b_i\forall i}\sum_i \frac{b_i-x_i+\delta S}{\gamma S}\cdot \left(1-n\delta- \frac{b_i-  y_i + \delta(b_i - x_i\land y_i)}{b_i-x_i + \delta S}\right)\\
	&\le \frac{S+n\delta S}{\gamma S}(1-n\delta) - \min_{y\colon y_i\ge b_i\forall i}\sum_i\frac{(b_i-  y_i)(1+\delta)}{\gamma S}\\
	&= \frac{1-(n\delta)^2}{\gamma} - \frac{1+\delta}{\gamma} \le 0,
	\end{align*}
	proving \eqref{eq:subSum}. Hence, using $\ell_i\in [-\eta,0]$,
	\begin{align*}
	\sum_i x_i' \cdot \partial_{x_i}D&\le - \beta(b_r-x_r)\left( 1-n\delta- \frac{[b_r-y_r]_+ + \delta(b_r-x_r\land y_r)}{b_r-x_r + \delta(b_r - x_r\land y_r)} - \delta\eta\right)\\
	&\le -\beta (b_r-x_r)\left(1-2n\delta\right) + \beta(1+\delta)[b_r-y_r]_+\qedhere
	\end{align*}
\end{proof}
The other quantity causing a change of $D$ is $b_r$. We have
\begin{align}
&\partial_{b_r}D=\nonumber \\
&\,\,\frac{\beta w_r}{\eta}\left[(1+\delta)\left(\log \frac{[b_r-y_r]_+ + \delta(b_r - x_r\land y_r)}{b_r-x_r + \delta(b_r - x_r\land y_r)} +\1_{y_r<b_r}\left(1 -\frac{b_r-y_r + \delta(b_r - x_r\land y_r)}{b_r-x_r + \delta(b_r - x_r\land y_r)}\right)\right) + 1\right].\label{eq:PhiChangeCr}
\end{align}
\paragraph{If $\boldsymbol{b_r-x_r\le 2\alpha}$.} Let us first establish \eqref{eq:wstarGoal} for the case $b_r-x_r\le 2\alpha$. Using the update rule $b_r'=\frac{\alpha}{w_r}$ for this case and the fact that $\log z\le z-1$ for all $z$, we get\footnote{Note that the second inequality below is very brittle as it requires that $\eta$ is no larger than what we defined it to be (whereas other inequalities also impose almost the same lower bound on $\eta$). However, even for larger $\eta$ one can prove that the algorithm still has $(1+\epsilon)$-competitive service cost by introducing a ``bad baseline'' potential $\sum_i w_i [y_i-b_i]_+$.}
\[
b_r'\cdot\partial_{b_r}D \le \frac{\alpha\beta}{\eta}\left[1_{y_r\ge b_r}(1+\delta)\log (\delta/(1+\delta)) + 1\right]
\le  \alpha\beta\left[-1_{y_r\ge b_r} + 1/\eta\right].
\]
Thus,
\begin{align*}
\alpha+D' &= \alpha+\sum_i x_i'\partial_{x_i}D + b_r'\partial_{b_r}D\\
&\le \alpha -\beta (b_r-x_r)\left(1-2n\delta\right) + \beta(1+\delta)[b_r-y_r]_+ + \alpha\beta\left[-1_{y_r\ge b_r} + 1/\eta\right].
\end{align*}
If $y_r\ge b_r$, then we obtain \eqref{eq:wstarGoal} via
\begin{align*}
\alpha+D' &\le \alpha - \alpha\beta\left[1- 1/\eta\right]\le 0\le (1+\epsilon)[\alpha+x_r-y_r]_+.
\end{align*}
If $y_r<b_r$, then we obtain \eqref{eq:wstarGoal} via
\begin{align*}
\alpha+D' &\le \alpha -\beta (b_r-x_r)\left(1-2n\delta\right) + \beta(1+\delta)(b_r-y_r) + \alpha\beta/\eta\\
&= \alpha +\beta (b_r-x_r)(2n+1)\delta + \beta(1+\delta)(x_r-y_r) + \alpha\beta/\eta\\
&\le \alpha \left(1 +2\beta (2n+1)\delta+ \beta/\eta\right) + \beta(1+\delta)(x_r-y_r) \\
&\le \alpha \beta + \beta(1+\delta)(x_r-y_r) \\
&\le (1+\epsilon)[\alpha + x_r-y_r]_+
\end{align*}
where the second inequality uses $b_r-x_r\le 2\alpha$, the third inequality uses $\beta=1+\frac{2}{\eta}$ and holds for $n$ sufficiently large, and the last inequality holds since $\epsilon\ge \beta(1+\delta)$.

\paragraph{If $\boldsymbol{b_r-x_r> 2\alpha}$.}For the case $b_r-x_r> 2\alpha$, the update rule $b_r'=-\frac{b_r-x_r}{2w_r}$ combined with \eqref{eq:PhiChangeCr} yields
\begin{align}
&b_r'\cdot\partial_{b_r}D =\nonumber\\
&\,\,-\frac{\beta (b_r-x_r)}{2\eta}\left[(1+\delta)\left(\log \frac{[b_r-y_r]_+ + \delta(b_r-x_r\land y_r)}{b_r-x_r + \delta(b_r-x_r\land y_r)} +\1_{y_r<b_r}\frac{y_r-x_r}{b_r-x_r + \delta(b_r-x_r\land y_r)}\right) + 1\right].\label{eq:PhiChangeCrCase2}
\end{align}
If $y_r\ge b_r$, then this simplifies to
\[
b_r'\cdot\partial_{b_r}D = \frac{\beta (b_r-x_r)}{2}\left[1 -1/\eta\right].
\]
Then we obtain \eqref{eq:wstarGoal} via
\begin{align*}
\alpha+D' &\le \alpha -\beta (b_r-x_r)\left(1-2n\delta\right)  + \frac{\beta (b_r-x_r)}{2}\left[1 -1/\eta\right]\\
&= \alpha -\frac{\beta(b_r-x_r)}{2}\left(1-4n\delta + 1/\eta\right)\le \alpha -\frac{b_r-x_r}{2}< 0,
\end{align*}
where the strict inequality uses $b_r-x_r > 2\alpha$. Thus 
$\alpha + D' \le (1+\epsilon)[\alpha+x_r-y_r]_+$.

Otherwise (if $y_r < b_r$), to show \eqref{eq:wstarGoal} we need to prove that
\begin{align}
\alpha&-\beta (b_r-x_r)\left(1-2n\delta\right) + \beta(1+\delta)(b_r-y_r) \nonumber\\
&-\frac{\beta (b_r-x_r)}{2\eta}\left[(1+\delta)\left(\log \frac{b_r-y_r + \delta(b_r-x_r\land y_r)}{b_r-x_r + \delta(b_r-x_r\land y_r)} +\frac{y_r-x_r}{b_r-x_r + \delta(b_r-x_r\land y_r)}\right) + 1\right]\nonumber\\
&\le (1+\epsilon)[\alpha+x_r-y_r]_+.\label{eq:wstarIntermediate}
\end{align}

If $y_r\in[x_r+\alpha,b_r)$, we need to show that the left-hand side is non-positive. Indeed, in this case the left-hand side is a convex function of $y_r$, so it is maximized when $y_r$ approaches one of the extreme values $b_r$ or $x_r+\alpha$. When $y_r$ approaches $b_r$ (from below), the expression for $b_r'\partial_{b_r}D$ is strictly smaller than the expression for the case $y_r=b_r$, so this case follows from the previous case. At the other extreme, $y_r=x_r+\alpha$, the left-hand side becomes
\begin{align*}
\alpha&-\beta (b_r-x_r)\left(1-2n\delta\right) + \beta(1+\delta)(b_r-x_r-\alpha) \\
&-\frac{\beta (b_r-x_r)}{2\eta}\left[(1+\delta)\left(\log \frac{(1+\delta)(b_r-x_r)-\alpha}{(1+ \delta)(b_r-x_r)} +\frac{\alpha}{(1+\delta)(b_r-x_r)}\right) + 1\right] \\
&\le \beta (b_r-x_r)\left((2n+1)\delta - \frac{1+\delta}{2\eta}\left(\log \frac{(1+\delta)(b_r-x_r)-\alpha}{(1+ \delta)(b_r-x_r)} +\frac{\alpha}{(1+\delta)(b_r-x_r)}\right) - \frac{1}{2\eta}\right) \\
&= \beta (b_r-x_r)\left((2n+1)\delta + \frac{1+\delta}{2\eta}\sum_{j=2}^\infty\frac{1}{j}\left(\frac{\alpha}{(1+\delta)(b_r-x_r)}\right)^j - \frac{1}{2\eta}\right) \\
&\le \beta (b_r-x_r)\left((2n+1)\delta + \frac{1}{2\eta}\sum_{j=2}^\infty\frac{1}{j}\left(\frac{1}{2}\right)^j - \frac{1}{2\eta}\right) \\
&\le 0,
\end{align*}
where the equality uses $\log(1-z)=-\sum_{j=1}^\infty\frac{z^j}{j}$ for $z\in[0,1)$ and $2\alpha< b_r-x_r$.

If $y_r\in[x_r,x_r+\alpha]$, then showing \eqref{eq:wstarIntermediate} still amounts to upper bounding a convex function of $y_r$, so only the extreme values are relevant. The case $y_r=x_r+\alpha$ was covered by the previous case, and $y_r=x_r$ will be included in the following only remaining case.

If $y_r\le x_r$, then the left-hand side of \eqref{eq:wstarIntermediate} is equal to
\begin{align*}
\alpha&+\beta(1+\delta)(x_r-y_r) \\
&+\beta (b_r-x_r)\left[(2n+1)\delta - \frac{1+\delta}{2\eta}\left(\log \frac{(1+\delta)(b_r-y_r)}{b_r-x_r + \delta(b_r-y_r)} +\frac{y_r-x_r}{b_r-x_r + \delta(b_r-y_r)}\right) -\frac{1}{2\eta}\right] \\
&\le \alpha+ \beta(1+\delta)(x_r-y_r) +\beta (b_r-x_r)\frac{1+\delta}{2\eta}\frac{x_r-y_r}{b_r-x_r + \delta(b_r-y_r)}\\
&\le \alpha+ \beta(1+\delta)(x_r-y_r) +\beta\frac{x_r-y_r}{2\eta}\\
&\le (1+\epsilon)(\alpha+x_r-y_r),
\end{align*}
where the last inequality uses $\epsilon\ge \frac{3}{\eta}$.
This concludes the proof of \eqref{eq:wstarIntermediate} and therefore \eqref{eq:wstarGoal}. We proved the following:
\begin{lemma}\label{lem:wstarService}
	The algorithm's service cost is $(1+\epsilon)$-competitive against the total offline cost.
\end{lemma}

\subsection{Movement cost analysis}
The proofs of Theorem~\ref{thm:refined} and Theorem~\ref{thm:non-refined} are completed by combining Lemma~\ref{lem:wstarService} with the following.
\begin{lemma}
The algorithm's movement cost is at most $O(\eta)$ times its service cost, up to a bounded additive error.
\end{lemma}
\begin{proof}
It suffices to bound the algorithm's movement cost for increasing the $x_i$, i.e., $\norm{[x']_+}=\sum_i w_i[x_i']_+$. Using the potential function $\Psi:=2\eta \sum_i w_i b_i$, we will show that $\norm{[x']_+} + \Psi'\le 4\eta\alpha$.

If $b_r-x_r\le 2\alpha$, then
\[
\norm{[x']_+} + \Psi' \le \eta(b_r-x_r) + 2\eta\alpha \le 4\eta\alpha.\]
Otherwise,
\[
\norm{[x']_+} + \Psi' \le \eta(b_r-x_r) - 2\eta \frac{b_r-x_r}{2} = 0 \le 4\eta\alpha.\qedhere\]
\end{proof}

\section{Non-convex metric allocation}\label{sec:nonconvex}
We now consider the variant of MAP where cost functions are allowed to be non-convex.

\paragraph{Notation.} Consider a tree with vertex set $V$ and fix some root $\rt\in V$. Write $V^0:=V\setminus\{\rt\}$. For $u\in V^0$, denote by $w_u>0$ the length of the edge $\{u,p(u)\}$ and denote by $p(u)$ the parent of $u$. We also write $v\prec u$ to denote that $v$ is a child of $u$. Let $\cL\subset V$ be the set of leaves of $V$. For $u\in V$, let $V_u$ be the set of vertices in the subtree rooted at $u$, and $\cL_u:=V_u\cap \cL$ the set of leaves in that subtree.  We take the metric space to be the set of leaves $\cL$ (without loss of generality, since one can add extra leaves at arbitrarily small distance from any internal vertex). The distances on $\cL$ are induced by the path metric.

Set $x_{\rt}:=1$. An allocation vector on $\cL$ corresponds to an element of the following convex body:
\begin{align*}
    K:=\{x\in \R^{V^0}\colon \forall u\in V\setminus\cL\colon x_u=\sum_{v\prec u} x_v\}
\end{align*}
As before (and as justified in Appendix~\ref{app:simplified}), we do not explicitly enforce that $x_i$ remain non-negative.

The norm to measure movement cost on a tree is given by $\norm{z}:= \sum_{u\in V^0} w_u |z_u|$ (i.e., the instantaneous movement cost of our algorithm will be given by $\norm{x'}$).

\subsection{Algorithm on trees}
Our algorithm is induced by mirror descent with the following \emph{weighted $\ell_2^2$ regularizer}:
\begin{align*}
    \Phi(x):=\frac{1}{2}\sum_{u\in V^0} w_u x_u^2.
\end{align*}
Denote by $\ell_t\in\cL$ the requested leaf at time $t$ and by $\alpha_t:=c_t(x(t))=c_{t,\ell_t}(x_{\ell_t}(t))$ the instantaneous service cost of the online algorithm at time $t$. We run mirror descent using the control function $f(t)=\alpha_t e_{\ell_t}$.\footnote{In order to satisfy the assumptions of \cite{BCLLM18} that guarantee that the mirror descent path is absolutely continuous and well-defined, we assume that $\alpha_t$ is a piecewise constant function of $t$. This is without loss of generality, as cost functions are approximated arbitrarily well by piecewise constant functions.} Observe that the algorithm does not need to know the entire cost function but only its value for the current online position.

Note that the Hessian of $\Phi$ is the diagonal matrix with entries $w_u$, and entries of the normal\footnote{In general, for a convex body defined by equality constraints, $K=\{x\in\R^N\colon Ax=b\}$ where $A\in\R^{m\times N}$ and $b\in\R^m$, the normal cone is given by $N_K(x)=\{A^{T}\lambda\colon \lambda\in\R^m\}$.} cone $N_K(x)$ for $x\in K$ are of the form $(\lambda_{p(u)}-\lambda_{u})_{u\in V^0}$ for $\lambda\in \R^{V}$, where $\lambda_{u}=0$ if $u\in\cL$ is a leaf. For $u\in V\setminus\cL$, we call $\lambda_u$ the \emph{Lagrange multiplier} of the constraint $x_u=\sum_{v\prec u} x_v$. Plugging into the mirror descent equation~\eqref{eq:MD} shows that the explicit dynamics of our algorithm is of the form
\begin{align*}
    x_u'(t)=\frac{\alpha_t\1_{u=\ell_t} + \lambda_u(t)-\lambda_{p(u)}(t)}{w_u}.
\end{align*}

\subsection{Basic properties of the algorithm}
By the following Lemma, which is proved in Appendix~\ref{app:LagrangeProof}, the Lagrange multipliers of internal vertices are positive:
\begin{lemma}\label{lem:lambdaPos}
For all $u\in V\setminus\cL$, $\lambda_u(t)> 0$.
\end{lemma}

The next lemma shows that the flow of the resource has $\ell_t$ as its unique sink.
\begin{lemma}\label{lem:flow}
For all $u\in V$, $x_u'(t)\ge 0$ if and only if $\ell_t\in\cL_u$.
\end{lemma}
\begin{proof}
    By Lemma~\ref{lem:lambdaPos} and the fact that $\lambda_v=0$ when $v$ is a leaf, we have $x_v'<0$ for all $v\in\cL\setminus\{\ell_t\}$. By the constraints of the polytope $K$, this shows that $x_v'<0$ for every $v\in V$ that is \emph{not} ancestor of $\ell_t$. Since the $x_{\rt}=1$ is fixed, this means that $x_u'>0$ for all $u\in V^0$ with $\ell_t\in\cL_u$.
\end{proof}

\subsection{Competitive analysis}
Unlike in the case of convex cost functions, the Bregman divergence is not useful as a potential function here. The reason is that the bound \eqref{eq:BregmanBound} is only useful if cost functions are sufficiently smooth, but in our case, the cost function can be an arbitrary non-increasing function. In particular, the only information we can use about the offline service cost is that it is at least $\alpha_t$ if $y_{\ell_t}\le x_{\ell_t}$, but otherwise it could be $0$ (even if $y_{\ell_t}$ is very close to $x_{\ell_t}$).

Instead, we use the following two potential functions:
\begin{align*}
    P(t)&:=\sum_{u\in V^0} w_u[x_u(t)-y_u(t)]_+\\
    Q(t)&:=-\sum_{u\in V^0} n_u w_u x_u(t),
\end{align*}
where $x,y\in K$ denote the state of the online resp. offline algorithm and $n_u=|\cL_u|$ is the number of leaves in the subtree rooted at $u$. The potential $P$ is a kind of ``one-sided matching'': If $[\cdot]_+$ were replaced by the absolute value, then it would capture the cost of moving from $x$ to $y$. The potential $Q$ is similar to the ``weighted depth potentials'' used in~\cite{BCLLM18,BCLL19}, where $n_u$ is replaced by some notion of ``depth'' of a vertex $u$. The overall potential function will be a weighted sum of $P$ and $Q$.

By the next Lemma, either the offline algorithm pays at least as much service cost as the online algorithm (when $x_{\ell_t}\ge y_{\ell_t}$; recall that cost functions are non-increasing) or the potential $P(t)$ decreases.
\begin{lemma}
If $y=y(t)$ is fixed at time $t$, then $P'(t)\le\alpha_t\1_{x_{\ell_t}(t)\ge y_{\ell_t}}-\lambda_{\rt}(t)$.
\end{lemma}
\begin{proof}
    Let $u_0\succ u_1\succ \dots\succ u_k$ be the vertices on the path from $u_0=\rt$ to $u_k=\ell_t$.
    
    If $x_{\ell_t}(t)\ge y_{\ell_t}$, then since $x_u(t)$ is increasing only for $u=u_j$ by Lemma~\ref{lem:flow}, we have
    \begin{align*}
        P'(t)&\le\sum_{j=1}^kw_{u_j}x_{u_j}'(t)
        = \alpha_t + \sum_{j=1}^k(\lambda_{u_j}(t)-\lambda_{u_{j-1}}(t))
        = \alpha_t - \lambda_{\rt}(t),
    \end{align*}
    where the last equation uses $u_0=\rt$ and $\lambda_{u_k}=0$ since $u_k=\ell_t\in\cL$.
    
    Otherwise (if $x_{\ell_t}(t)< y_{\ell_t}$), let $h$ be maximal such that $x_{u_h}(t)\ge y_{u_h}$. Note that $h$ exists since $x_{\rt}=y_{\rt}=1$.
    and $h<k$ by the assumption $x_{\ell_t}(t)<y_{\ell_t}$. Since $x_{u_h}(t)\ge y_{u_h}$ but $x_{u_{h+1}}(t)< y_{u_{h+1}}$, we can find inductively a sequence $\hat u_h\succ \hat u_{h+1}\succ\dots\succ \hat u_{m}$ with $\hat u_h=u_h$ and $\hat u_m\in\cL$ such that for all $j=h+1,\dots,m$, we have $x_{\hat u_j}(t)>y_{\hat u_j}$ and $\ell_t\notin \cL_{\hat u_j}$.
    
    Then
    \begin{align*}
        P'(t)&=\sum_{u\in V^0}w_u\partial_t[x_u(t)-y_u]_+\\
        &\le\sum_{j=1}^hw_{u_j}\partial_t[x_{u_j}(t)-y_{u_j}]_+ + \sum_{j=h+1}^m w_{\hat u_j}\partial_t[x_{\hat u_j}(t)-y_{\hat u_j}]_+\\
        &\le\sum_{j=1}^hw_{u_j}x_{u_j}'(t) + \sum_{j=h+1}^m w_{\hat u_j}x_{\hat u_j}'(t)\\
        &= \sum_{j=1}^h(\lambda_{u_j}(t)-\lambda_{u_{j-1}}(t)) + \sum_{j=h+1}^m(\lambda_{\hat u_j}(t)-\lambda_{\hat u_{j-1}}(t))\\
        &= \lambda_{u_h}(t)-\lambda_{u_0}(t) + \lambda_{\hat u_m}(t) - \lambda_{\hat u_h}(t)\\
        &= -\lambda_{\rt}(t).
    \end{align*}
    The first inequality holds because the omitted summands for $u\in V^0$ with $\ell_t\in V_u$ are precisely those for $u=u_j$ for $j\ge h+1$, which are $0$ because $x_{u_j}(t)<y_{u_j}$ by maximality of $h$, and the omitted summands for $u\in V^0$ with $\ell_t\notin V_u$ are non-positive by Lemma~\ref{lem:flow}.
    The second inequality uses that $x_{u_j}'(t)\ge 0$ for $j=0,\dots,h$ by Lemma~\ref{lem:flow}, and that $x_{\hat u_j}(t)>y_{\hat u_j}$ for $j=h+1,\dots,m$.
    The last equation uses that $u_h=\hat u_h$, $u_0=\rt$ and $\lambda_{\hat u_m}(t)=0$ since $\hat u_m\in\cL$.
\end{proof}

The next lemma describes the change of the potential $Q$.

\begin{lemma}
$Q'(t)=n\lambda_{\rt}(t)-\alpha_t$.
\end{lemma}
\begin{proof}
    We have
    \begin{align*}
        Q'(t)&= - \sum_{u\in V^0} n_u(\alpha_t\1_{u=\ell_t} + \lambda_{u}(t)-\lambda_{p(u)}(t))\\
        &= -\alpha_t + n\lambda_{\rt}(t) - \sum_{u\in V^0\setminus\cL}\lambda_u(t)\Big(n_u - \sum_{v\prec u} n_v \Big)\\
        &= -\alpha_t+n\lambda_{\rt}(t),
    \end{align*}
    where the second equation uses $n_{\ell_t}=1$, $n_{\rt}=n$ and $\lambda_u=0$ for $u\in\cL$.
\end{proof}

\medskip
We are now ready to prove the tight upper bound on trees.

\medskip

\begin{proof}[Proof of~Theorem~\ref{thm:nonconvexTrees}]
Since the total movement cost due to increasing variables $x_u$ is equal to the cost due to decreasing variables $x_u$ (up to a bounded additive term), it suffices to consider the increasing movement cost, i.e., $\sum_{u\in V^0}w_u [x_u'(t)]_+$. By Lemma~\ref{lem:flow}, we know that this cost is only suffered at ancestors of the root. Thus,
the algorithm's instantaneous movement cost at time $t$ is given by
\[    \sum_{u\in V^0\colon \ell_t\in \cL_u} w_u x_u'(t) = 
    \alpha_t + \sum_{u\in V^0\colon \ell_t\in \cL_u} (\lambda_u(t)-\lambda_{p(u)}(t))= \alpha_t - \lambda_{\rt}(t).
\]
Since the service cost is $\alpha_t$, the total instantaneous online cost is $\text{On}'=2\alpha_t-\lambda_{\rt}(t)$. Thus, using $(2n-1)P + 2Q$ as the overall potential function, we get
\begin{align*}
    \text{On}' + (2n-1)P' + 2Q' &= 2\alpha_t - \lambda_{\rt}(t) + (2n-1)\left(\alpha_t\1_{x_{\ell_t}(t)\ge y_{\ell_t}}-\lambda_{\rt}(t)\right) + 2n\lambda_{\rt}(t) - 2\alpha_t\\
    &=(2n-1)\alpha_t\1_{x_{\ell_t}(t)\ge y_{\ell_t}}
\end{align*}
when the offline algorithm is at $y$ (and does not move). In this case, the offline service cost is at least $\alpha_t\1_{x_{\ell_t}(t)\ge y_{\ell_t}}$ since cost functions are non-increasing.

When the offline algorithm moves, clearly the potential $(2n-1)P + 2Q$ increases by at most $O(n)$ times the offline movement cost. Thus, the algorithm is $O(n)$-competitive.
\end{proof}

\subsection{General metrics}

It is well-known~\cite{FRT04} that any $n$-point metric space $(M,d)$ can be randomly embedded into a tree metric with distortion at most an $O(\log n)$ factor. Similarly, $(M,d)$ can be embedded deterministically into the metric $d_T$ induced by a minimum spanning tree $T$ of $(M,d)$, and it then holds that $d(i,j)\le d_T(i,j)\le (n-1)d(i,j)$ for all $i,j\in M$.\footnote{To see the latter inequality, observe that $i$ and $j$ are connected by a path of at most $n-1$ edges in $T$, each of which has length at most $d(i,j)$ since $T$ is a minimum spanning tree.} Embedding the metric space into a tree metric in this way and running our $O(n)$-competitive algorithm for trees, we obtain Corollary~\ref{cor:nonconvexGeneral}.

\subsection{Lower bound}
We now show that every algorithm for MAP with non-convex cost functions has competitive ratio at least $\Omega(n)$ on any metric space, proving Theorem~\ref{thm:LB}.

Let $M$ be an $n$-point metric space and assume without loss of generality that the minimal distance in $M$ is $1$. Consider the instance where at any time $t$, we have $r_t = \arg\min_{i\in M} x_i(t-1)$ and the cost function $c_{t,r_t}$ is defined as 
\begin{equation*}
    c_{t,r_t}(z) = 
    \begin{cases}
    1/n^2 &\text{ if $z < 1/(n-1)$}\\
    0 &\text{ otherwise.}
    \end{cases}
\end{equation*}
Since the algorithm must maintain $\sum_i x_i(t-1)=1$, we must have $x_{r_t}(t-1) \leq 1/n$. Now if the online algorithm increases $x_{r_t}$ to at least $1/(n-1)$, then it incurs movement cost of at least $1/(n-1) - 1/n \geq 1/n^2$. Otherwise, it incurs service cost of $1/n^2$.

Denote by $i^*\in M$ the point that is requested least often over the entire request sequence. Consider the offline algorithm that stays put in the configuration $y(t) = y$ with $y_{i^*}=0$ and $y_{i}=1/(n-1)$ for $i\ne i^*$. This algorithm does not incur any movement cost (except for a fixed constant amount at the beginning to reach this configuration), and it pays service cost $1/n^2$ only for requests at $i^*$. By definition, since $i^*$ is the least requested point, it incurs this cost of $1/n^2$ on at most a $1/n$ fraction of time steps. As the online algorithm incurs at least $1/n^2$ cost at each time step, the lower bound follows.
\section{Conclusion}\label{sec:conclusion}

An intriguing open problem is whether there exists a polylog$(n)$-competitive algorithm for MAP on general metrics. The concepts and techniques introduced in this paper are an important step, but some additional ideas are still needed to resolve this question on HSTs (and thus general metrics).

We feel that our techniques to better understand scale-freeness may also be useful for improving the current competitive ratio bounds for randomized $k$-server, as we explain below.

A key question in designing an algorithm on HSTs that solves a weighted-star instance of the problem at its internal vertices is to find a suitable definition of ``service cost'' that should be felt at the internal vertices of the HST. For MTS, this is now well-understood and algorithms with a competitive ratio of $O(\log n)$ on HSTs were obtained recently~\cite{BCLL19,CoesterL19}.

For $k$-server, the decision at an internal vertex is essentially ``How many servers should there be in this subtree?''. From the request sequence it may be clear that the number of servers in a subtree should be within some sub-interval $[a,b]\subsetneq[0,k]$, and the decision to be made is only where inside $[a,b]$ the algorithm's number of servers in the subtree should lie. Note that this is essentially a question of scale-freeness, and passing a suitable convex cost function to the internal vertices of the tree should help the algorithm to find the right scale and region similarly to MAP on a weighted star. Indeed, the $O(\log^2k)$-competitive algorithm for $k$-server on HSTs~\cite{BCLLM18} is based on running mirror descent on a sophisticated polytope, which impicitly helps the algorithm to identify similar regions.

We believe that a suitable extension of our algorithm for weighted stars will allow to achieve an $O(\log n)$ competitive ratio for MAP on HSTs (similarly to MTS), and as a corollary also reduce the competitive ratio of $k$-server on HSTs to a single logarithm.
\section*{Acknowledgments}
We thank Ravi Kumar, Manish Purohit and Erik Vee for many useful discussions that inspired this work.

Part of this work was carried out while both authors were at CWI in Amsterdam. Nikhil Bansal is supported  by the NWO VICI grant 639.023.812. Christian Coester is supported by the Israel Academy of Sciences and Humanities \& Council for Higher Education Excellence Fellowship Program for International Postdoctoral Researchers.

\newpage
\appendix
\section{Simplified continuous-time model}
\label{app:simplified}

\paragraph{Reduction to continuous-time model.} For the reduction from the discrete to the continuous time models, the cost function revealed at discrete time $t$ will be revealed at all times in $(t-1,t]$ in the continuous setting. Let $t^*=\arg\min_{t^*\in(t-1,t]} c_t(x(t^*))$. Then the algorithm in the discrete time version goes to state $x(t^*)$ at time $t$. By the triangle inequality, this can only reduce the algorithm's cost, while the offline cost does not change.

\paragraph{Cost at a single location.}
Consider a cost function $c(x)=\sum_{i\in M}c_i(x_i)$. Let $c^{(i)}$ be the cost function given by $c^{(i)}(x)=\frac{c_i(x_i)}{N}$, for some large $N$. Instead of issuing $c$, we can issue the cost functions $c^{(i)}$ for each point $i\in M$ in cyclic order, repeating $N$ times. As $N\to\infty$, any competitive algorithm must eventually converge towards some point $x\in\Delta$ to avoid paying infinite movement cost. Thus, the costs of the online and offline algorithm are unaffected by this change.

\paragraph{Linearizing the cost function.}

We now argue that we can assume each $c_{t,r_t}$ to be
linearly decreasing and truncated at $0$ (i.e.~a
hinge loss).
Indeed, given some convex function $c_{t,r_t}$, consider the function $\tilde{c}_{t,r_t}$ that is tangent to the graph of $c_{t,r_t}$ at the point $(x_{r_t}(t),c_{t,r_t}(x_{r_t}(t)))$ and truncated at value $0$.
As $\tilde{c}_{t,r_t}(x_{r_t}(t))=c_{t,r_t}(x_{r_t}(t))$, the instantaneous service cost of the online algorithm does not change, while on the other hand, as $\tilde{c}_{t,r_t}(y_{r_t})\le c_{t,r_t}(y_{r_t})$ for all $y_{r_t}$ (by convexity of $c_{t,r_t}$), the service cost of the offline algorithm can only decrease.

We can also assume that the slope of this function (where it is non-zero) is exactly $-1$: If the slope were $-s\ne -1$, we can replace $\tilde{c}_{t,r_t}$ by $\tilde{c}_{t,r_t} / s$ to get a function with slope $-1$, and scale time so that this function is issued for $s$ times as long as $\tilde{c}_{t,r_t}$ would be issued. This does not affect the service cost incurred by any algorithm. Finally, we can also assume that the service cost becomes $0$ at some $s_t\in[0,1]$. Indeed, additively reducing the cost function so that this becomes the case only hurts the online algorithm as an additive reduction in cost for both the online and offline algorithm only increases the competitive ratio.

\paragraph{Allowing negative $\boldsymbol{x_i}$.} Any algorithm that decreases $x_i$ to a negative value can be turned into an algorithm that keeps $x_i$ non-negative without increasing its competitive ratio. To achieve this, we can replace every cost function $c_{t,i}$ by $\tilde c_{t,i}$ defined via $\tilde c_{t,i}(x_i)=c_{t,i}(x_i)+a[-x_i]_+$, for some large $a>0$. As $a\to\infty$, the configuration of any competitive algorithm converges to an $x$ with $x_i\ge 0$ for all $i$.

\section{Equivalence of non-increasing and non-decreasing costs}\label{app:incrDecr}
Recall that we defined cost functions to be non-increasing, but in MTS they are non-decreasing. The following lemma shows that these two models are equivalent.

\begin{lemma}\label{lem:increasingDecreasing}
    For any metric space $M$ and $\rho\ge 1$, there is a $\rho$-competitive algorithm for MAP if and only if there is a $\rho$-competitive algorithm for the variant of the problem with non-decreasing cost functions.
\end{lemma}
\begin{proof}
Suppose there is a $\rho$-competitive algorithm for MAP (with non-increasing cost functions). To serve an instance $I$ with non-decreasing cost functions, we will construct (online) an instance $\hat{I}$ with non-increasing cost functions. Denote by $x$ the allocation vector for instance $I$ and by $\hat x$ the allocation vector for $\hat I$. We will maintain the correspondence $\hat x_i=\frac{1-x_i}{n-1}$.

Assume without loss of generality (by arguments analogous to those in Appendix~\ref{app:simplified}) that every cost function in $I$ is of the form $c(x)=(x_r-s)_+$ for some $s\in[0,1]$ and $r\in M$. When such a cost function arrives in $I$, we issue in $\hat I$ the cost function
 \begin{align*}
    \hat c(\hat x):=
    \left(\frac{1-s}{n-1} - \hat x_r\right)_+.
\end{align*}

When the algorithm in instance $\hat I$ moves to $\hat x$, then in the original instance $I$ we move to the allocation vector $x$ given by $x_i=1-(n-1)\hat x_i$.

By the correspondence between $\hat x$ and $x$, both the service and movement costs in instance $\hat I$ are a factor $n-1$ smaller than in instance $I$, and the same is true for the offline algorithm by using an analogous correspondence. Thus, the competitive ratio is the same.

Note that the correspondence ensures that $\sum_i x_i=1$ is equivalent to $\sum_i \hat x_i=1$. Moreover, $x_i\ge 0$ is maintained because the algorithm for $\hat I$ has no reason to increase $\hat x_i$ beyond $\frac{1}{n-1}$ as the service cost is $0$ already when $\hat x_r\ge 0$ (any algorithm that increases $\hat x_i$ beyond $\frac{1}{n-1}$ can be made lazy by delaying the movement of mass, which only decreases its cost).

The converse reduction is  analogous, but in that case it can happen that $x_i\ge 0$ will be violated. However, as argued in Appendix~\ref{app:simplified}, such an algorithm can then be converted into one that maintains $x_i\ge 0$.
\end{proof}

\section{Derivation of mirror descent dynamics}\label{app:MD}
Consider the regularizer defined by
\begin{align}\Phi(x):= \frac{1}{\eta}\sum_i w_i(b_i-x_i+\delta S)\log\left(\frac{b_i-x_i}{S}+\delta\right).
\end{align}
In contrast to the definition~\eqref{eq:regWighted}, we omitted the $[\cdot]_+$ around $b_i-x_i$ inside the logarithm because the quantity is always positive for the online algorithm $x$ anyway. By computing partial derivatives with respect to the $x_i$, it is straightforward to calculate that the Hessian $\nabla^2\Phi(x)$ is a diagonal matrix, with diagonal entries
\begin{align*}
    \partial_{x_ix_i}\Phi(x)=\frac{w_i}{\eta(b_i-x_i+\delta S)}
\end{align*}
Moreover, for any $x$ in $\Delta=\{x\in\R^M\colon \sum_i x_i=1\}$, the normal cone is given by
\begin{align*}
    N_\Delta(x)=\{\lambda\cdot\1\colon \lambda\in\R\},
\end{align*}
where $\1\in\R^M$ is the all-$1$ vector. Thus, for a control function of the form $f(t)=a_t e_{r_t}$, the mirror descent equation \eqref{eq:MD} yields (omitting the dependence on $t$ in the notation)
\begin{align*}
    x_i'=\frac{\eta}{w_i}(b_i-x_i+\delta S)(a\1_{i=r} - \lambda)
\end{align*}
for some $\lambda\in\R$. Since $\sum_{i\in M} x_i=1$ is constant, we must have
\begin{align*}
    0=\sum_{i\in M}x_i'=\eta\sum_i \frac{b_i-x_i+\delta S}{w_i}(a\1_{i=r} - \lambda),
\end{align*}
which shows that
\begin{align*}
    \lambda = \frac{a(b_r-x_r+\delta S)}{w_r\sum_i \frac{b_i-x_i+\delta S}{w_i}} = \frac{a(b_r-x_r+\delta S)}{\gamma w_rS}
\end{align*}
for $\gamma=\sum_i \frac{b_i-x_i+\delta S}{Sw_i}$. We can therefore write the update rule for $x_i$ as
\begin{align*}
    x_i'&=a\frac{\eta}{w_i}(b_i-x_i+\delta S)\left(\1_{i=r} - \frac{b_r-x_r+\delta S}{\gamma w_rS}\right)\\
    &=\eta a\frac{b_r-x_r+\delta S}{w_r}\left(\1_{i=r} - \frac{b_i-x_i+\delta S}{\gamma w_iS}\right)
\end{align*}
For $a=\frac{b_r-x_r}{b_r-x_r+\delta S}$, we obtain the update rule~\eqref{eq:xUpdateNew}. For $\eta=w_i=1$, $\delta=\frac{1}{n}$ and $a=\frac{\alpha}{b_r-x_r+\delta S}$, we obtain the update rule~\eqref{eq:xupdate}.

\section{Proof of Lemma~\ref{lem:lambdaPos}}\label{app:LagrangeProof}

For convenience we drop $t$ from the notation.

We first show that for all $u\in V\setminus\cL$, the sign of $\lambda_u$ and $\lambda_{p(u)}$ is the same. We show this first for all $u$ with $\ell_t\notin \cL_u$ and then for $u$ with $\ell_t\in \cL_u$.

Suppose there exists $u\in V\setminus\cL$ with $\ell_t\notin \cL_u$ such that $\lambda_u$ and $\lambda_{p(u)}$ have different signs, and choose such $u$ of \emph{maximal} depth. Let $s\in\{-1,0,1\}$ be the sign of $\lambda_u$. By maximality of the depth of $u$, for all $v\in\cL_u$, $\lambda_{p(v)}$ has sign $s$. Since $\lambda_v=0$ for any leaf $v$, this shows that $x_v'$ has sign $-s$ (since $v\ne\ell_t$). Since $x_u'=\sum_{v\in\cL_u}x_v'$, also $x_u'$ has sign $-s$. But since $\lambda_u$ has sign $s$ and $\lambda_{p(u)}$ has a different sign, this is not possible.

Suppose now that there exists $u\in V\setminus\cL$ with $\ell_t\in \cL_u$ such that $\lambda_u$ and $\lambda_{p(u)}$ have different signs, and choose such $u$ of \emph{minimal} depth. Let $s\in\{-1,0,1\}$ be the sign of $\lambda_{p(u)}$. By minimality of the depth of $u$, and since we just showed that $\lambda_v$ and $\lambda_{p(v)}$ have the same sign if $v$ is \emph{not} an ancestor of $\ell_t$, we conclude that $\lambda_v$ has sign $s$ for all $v\in V\setminus(V_u\cup\cL)$. So for any leaf $v\in \cL\setminus\cL_u$, the sign of $\lambda_{p(v)}$ is $s$ and thus the sign of $x_v'$ is $-s$ (using that $\ell_t\in\cL_u$ and $\lambda_v=0$ for any leaf $v$). Since the total leaf mass does not change (i.e., $\sum_{v\in\cL} x_v'=0$), for the remaining leaves it holds that $\sum_{v\in\cL_u} x_v'$ must have sign $s$, and thus $x_u'$ has sign $s$. But since $\lambda_{p(u)}$ also has sign $s$ and $\lambda_u$ has a different sign, this is not possible.

So all $\lambda_u$ for $u\in V\setminus\cL$ have the same sign. This sign can only be positive since otherwise we would have $\sum_{v\in\cL} x_v'>0$ (recalling again that for $v\in\cL$, $\lambda_v=0$).

\bibliographystyle{alpha}
\bibliography{bibliography}

\end{document}